\documentclass[a4paper, reqno]{amsart}

\usepackage[utf8]{inputenc} 
\usepackage{amsmath, amssymb, amsthm}
\usepackage{mathrsfs}
\usepackage{mathtools}
\usepackage{dsfont}
\usepackage{array}
\usepackage{booktabs}
\usepackage{microtype}
\usepackage{subcaption}
\usepackage{pgfplots}
\usepackage{tikz-network}
\usepackage{url}
\usepackage[giveninits, maxbibnames=99, backend=biber, eprint=false, doi=false, isbn=false]{biblatex}
\usepackage[hyperfootnotes=false]{hyperref}

\graphicspath{{figures/}}

\pgfplotsset{compat=newest}
\usetikzlibrary{arrows.meta}

\newcommand\tikzgraphsettings{\tikzset{
	every node/.style = {
		circle, fill,
		inner sep=0pt,
		minimum size=1.7mm,
		outer sep=1pt,
		auto
	},
	every label/.append style={font=\footnotesize},
	every mark/.append style={
		mark size=2.25pt
	},
	empty/.style = {
		font = \small,
		rectangle,
		inner sep=0pt,
		draw=none,
		fill=none
	},
	>=latex}
}

\SetVertexStyle[LineColor=white, MinSize=.1, FillColor=black, TextFont=\small]    
\SetEdgeStyle[LineWidth=.4pt]

\colorlet{mplblue}{black}
\colorlet{mplorange}{black}
\colorlet{mplgreen}{black}
\colorlet{mplred}{black}
\colorlet{mplpurple}{black}
\colorlet{mplbrown}{black}
\colorlet{mplgray}{black}
\colorlet{mplyellowgreen}{black}
\colorlet{mplcyan}{black}

\theoremstyle{plain}
\newtheorem{theorem}{Theorem}
\newtheorem*{theorem*}{Theorem}
\newtheorem{proposition}{Proposition}
\newtheorem{lemma}{Lemma}
\newtheorem{corollary}{Corollary}
\theoremstyle{definition}
\newtheorem{definition}{Definition}
\newtheorem{example}{Example}
\newtheorem{axiom}{Axiom}
\theoremstyle{remark}
\newtheorem{remark}{Remark}

\DeclareMathOperator{\id}{Id}
\newcommand{\de}{\ensuremath{\mathrm{d}}}

\newcommand\1{\mathds{1}}

\newcommand\restr[2]{{\left.\kern-\nulldelimiterspace#1\mathchoice{\vphantom{\big|}}{}{}{}\right|_{#2}}}

\title[Modeling advection on directed networks]{Modeling advection on distance-weighted directed networks}

\author{Michele Benzi}
\address{Scuola Normale Superiore. Pisa, PI 56126, Italy.}
\email{michele.benzi@sns.it}

\author{Fabio Durastante}
\address{Università di Pisa. Pisa, PI 56126, Italy.}
\email{fabio.durastante@unipi.it}

\author{Francesco Zigliotto}
\address{Scuola Normale Superiore. Pisa, PI 56126, Italy.}
\email[Corresponding author]{francesco.zigliotto@sns.it}


\begin{document}

\begin{abstract}
In this paper we propose a model for describing advection dynamics on distance-weighted directed graphs. To this end we establish a set of key properties, or axioms, that a discrete advection operator should satisfy, and prove that there exists an essentially unique operator satisfying  all such properties. Both infinite and finite networks are considered, as well as possible variants and extensions. We illustrate the proposed model through examples, both analytical and numerical, and we describe an application to the simulation of a traffic network.
\end{abstract}

\maketitle
\thispagestyle{empty}

\section{Introduction}

Transport phenomena, encompassing the movement of particles, energy, or information, play a crucial role in numerous scientific and engineering domains. Traditionally, these dynamics have been analyzed within the context of continuous media using tools from differential equations~\cite{MR2273657} and continuum mechanics~\cite{ContMech}. However, many systems of interest, such as social networks, biological networks~\hbox{\cite{BellomoNetworkCovid,AguiarContagion}}, and communication infrastructures, are inherently discrete and are best described using combinatorial graphs~\cite{BarabasiReview,StructureAndDynamics}. This observation motivates our interest in developing a rigorous framework to study transport processes within these discrete structures.

Combinatorial graphs, consisting of nodes connected by directed and weighted edges, are one of the most utilized representations of complex networks. In this discrete setting, we aim to rigorously model transport phenomena involving the movement or flow of quantities associated with the nodes along the edges of the graph, while capturing the essential properties of their continuous counterparts. Specifically, we isolate a set of properties, or axioms, that capture the essence of the physical phenomenon of transport on a network in order to construct a robust theoretical framework grounded in fundamental principles that describe transport processes in discrete systems. This axiomatic approach allows us to uniquely identify a discrete advection operator for distance-weighted directed networks,
and provides a systematic basis for analyzing and understanding transport dynamics on graphs.

Our proposed study complements the analysis of diffusion phenomena on complex networks~\cite{Lambiotte}. Diffusion analysis characterizes the inherent structure and topology of the underlying network, extracting information about diffusion patterns, identifying central \cite{PhysRevResearch.2.033104} and critical nodes for controlling the spread of quantities, designing interventions to achieve a prescribed network state, and recognizing community structures \cite{doi:10.1073/pnas.0903215107}. These studies are often well described by the use of the discrete Laplacian matrix or its variants for both undirected and directed networks, including non-local phenomena \cite{MR4130854,MR4364780,PhysRevE.90.032809}. By considering transport phenomena of different types, we can describe additional network properties and focus on directional and conservative flows, thereby broadening our understanding of network dynamics and enhancing our ability to model and optimize real-world systems.

The paper is organized as follows: in Sections~\ref{sec:related_works} and~\ref{sec:notation}, we review previous definitions of transport-like operators on combinatorial graphs and introduce the necessary notation and formal setting. Sections~\ref{sec:axiomatic_construction} formalizes the properties that an advection operator on a combinatorial graph should possess, introducing our new axiomatic framework and investigating the newly defined dynamic. Within this framework, we prove the existence and uniqueness of the operator in Section~\ref{sec:characterization}. Section~\ref{sec:advection_dynamics} presents examples on finite and infinite graphs, illustrating the chosen axioms and their implications, including applications to road networks. Finally, in Section~\ref{sec:conclusion}, we draw conclusions and suggest possible future developments and applications for the new operator.

\subsection{Related works}
\label{sec:related_works}

In this section, we provide a brief review of contributions related to the advection operator on graphs.

One of the earliest approaches involves studying the dynamics of networked multi-agent systems utilizing an advection-based coordination algorithm \cite{chapman2011,rak2017}. The authors highlight that while consensus dynamics are typically modeled using the discretized diffusion equation \cite{Lambiotte}, advection dynamics offer a complementary approach where the sum of states is conserved. This conservation property makes advection-based algorithms particularly appealing for applications such as formation control and sensor coverage. The authors introduce a formulation of advection dynamics on directed graphs via the continuous transport equation
\[
\dfrac{\partial f}{\partial t} = -\nabla \cdot (\omega f),
\]
where \(\omega\) is a velocity field, and $f$ the density to be transported. In the discrete case, it is assumed that each edge \((u,v)\) of the underlying graph is associated with a positive velocity \(\omega_{uv}\), leading to the discrete equation
\[
\dfrac{\de }{\de t} f_t(u) = -\bigl[A_G f_t\bigr](u)
\]
for any node $u$, where
\[
[A_G f_t](u) = \Biggl[\sum_{w \in N^+(u)} \omega_{uw}\Biggr] f_t(u) - \sum_{v \in N^-(u)} \omega_{vu} f_t(v).
\]

An alternative construction for undirected graphs has also been considered \cite{MR4469494}. The advection operator introduced here leverages the node degrees to define the preferred direction of movement within the network, making it solely dependent on the graph's topological structure. This work extends the operator to an advection-diffusion equation by combining the new advection operator with the standard graph Laplacian. Additionally, the convergence properties of systems governed by advection-diffusion equations on graphs are discussed, demonstrating that the degree-biased advection operator results in dynamics that strongly depend on the network structure, with specific transient dynamics influenced by the non-zero eigenvalues of the operator \cite{MR4469494}.

A tangentially related contribution involves the study of solutions to a discrete Navier-Stokes equation on combinatorial graphs \cite{MR2719791}. The authors propose an adaptation of the equation to finite, connected, weighted, and regular graphs with the aim of obtaining an ordinary differential equation whose solutions correspond to discrete conservation laws on the graphs. To define the gradient operator needed for the discrete reformulation of the Navier-Stokes equation, it is necessary that each node of the graph has the same number of incident edges, allowing the orientation of the motion to be fixed by choosing a consistent permutation that assigns the direction of the flow for each node.

\subsection{Notation and fundamentals on combinatorial graphs}
\label{sec:notation}

Here, we provide the basic definitions and notations that will be used throughout the rest of the article. 

\begin{definition}[Graph]
A \emph{directed and weighted graph} \( G \) is defined as a triplet \( G = (V, E, \omega) \) where:
\begin{itemize}
    \item \( V = \{v_1,v_2,\ldots,v_k,\ldots\} \) is a finite or countable set of nodes.
    \item \( E \) is a set of directed edges $e$, where each edge is an ordered pair \( e = (u, v) \) with \( u, v \in V \).
    \item \( \omega: E \rightarrow \mathbb{R}_+ \) is a weight function that assigns a positive real-valued weight to each edge in \( E \).
\end{itemize}
For any $u \in V$ we denote by $\deg(u)$ the number of edges $e \in E$ having $u$ either as first or second component. We denote with $\deg^+(u)$ the number of edges $e \in E$ having $u$ as first component, and with $\deg^-(u)$ the number of the ones having it as last component.
\end{definition}

\begin{definition}[Oriented graph]
An oriented graph is a directed graph with no bidirectional edges, i.e., for every pair of nodes $u, v \in V$, if $(u, v) \in E$, then $(v, u) \notin E$.
\end{definition}

\begin{definition}[Walk, path and cycle]
A walk in a graph $G = (V, E)$ is a sequence of non necessarily distinct nodes $(v_1, v_2, \ldots, v_k)$ such that for each consecutive pair of nodes $(v_i, v_{i+1})$ with $1 \leq i < k$, there is an edge $(v_i, v_{i+1}) \in E$. We call cycle a walk for which the starting and ending node coincide, while all the others are distinct. A graph without cycles is called acyclic. If for any two nodes in~$V$ there exists a walk having them as endpoints, then we say that $G$ is (strongly) connected.
\end{definition}

\begin{definition}[Oriented tree]\label{def:oriented-tree}
Let $G$ be an oriented graph, and let $\hat{G}$ be the graph constructed from $G$ by reciprocating each edge. If $\hat{G}$ is connected and does not contain cycles, then we call $G$ an oriented tree.
\end{definition}

\begin{definition}[Neighbourhood]\label{def:neighbourhood}
Given a graph $G = (V,E,\omega)$, we denote by $N(u)$ the neighborhood of the node $u\in V$, i.e., the set of nodes that are connected to $u$ by an edge in either direction. We also introduce the notation 
\[
N[u]  = \{u\} \cup N(u)
\]
for closed neighborhood and, more generally,
\[
\begin{cases}
N_{k+1}[u] = \displaystyle \bigcup_{v\in N_k[u]}N[v], & \forall\,k \ge 0\\
N_0[u]=\{u\}.
\end{cases}
\]
The $k$-iterated neighbor of $u$ can be then defined as follows:
\[
N_k(u) =
\begin{cases}
N_k[u]\setminus N_{k-1}[u] & \text{for $k\ge1$}\\
\{u\} & \text{for $k=0$.}
\end{cases}
\]
In some cases, it makes sense to consider only the successor nodes starting from a given node $u$, that is, only those pointed to by an arc of the form $(u,v)$. In this case, we denote the quantities as $N^{+}(u)$, $N^{+}[u]$ and $N_k^{+}[u]$ respectively. We can further consider the entire cone of successors of $u$ and denote it by $N_\infty^+[u]$. The analogous concepts based on predecessor nodes ($N^{-}(u)$, $N^{-}[u]$, etc.) are similarly defined. See Figure~\ref{fig:neighbourhood} for a graphical representation.

\newcommand\NeighborVertex[2][1]{%
	\SetVertexStyle[LineColor=black, MinSize=.1, FillColor=black, TextFont=\small] 
    \Vertex[color=white,shape=diamond,#1]{#2}%
	\SetVertexStyle[LineColor=white, MinSize=.1, FillColor=black, TextFont=\small]     
}

\SetVertexStyle[LineColor=white, MinSize=.1, FillColor=black, TextFont=\small] 

\begin{figure}
    \def\neighborcolor{red}
    \centering
    \begin{subfigure}[b]{.15\columnwidth}
    \begin{tikzpicture}
    \NeighborVertex[x=0,y=0,size=.1]{pl}
    \NeighborVertex[x=0,y=1]{pu}
    \Vertex[x=0.5,y=0.5,label={$u$},position=above,distance=.5mm]{u}
    \NeighborVertex[x=1,y=1]{su}
    \Vertex[x=1,y=0]{sl}
    \Vertex[x=1.5,y=0.5]{v}
    \Edge[Direct](pl)(u)
    \Edge[Direct](pu)(u)
    \Edge[Direct](u)(su)
    \Edge[Direct](su)(v)
    \Edge[Direct](sl)(v)
    \end{tikzpicture}
    \caption{$N(u)$}
    \end{subfigure}
    \quad
    \begin{subfigure}[b]{.15\columnwidth}
    \begin{tikzpicture}
    \NeighborVertex[x=0,y=0]{pl}
    \NeighborVertex[x=0,y=1]{pu}
    \NeighborVertex[x=0.5,y=0.5,label={$u$},position=above,distance=.5mm]{u}
    \NeighborVertex[x=1,y=1]{su}
    \Vertex[x=1,y=0]{sl}
    \NeighborVertex[x=1.5,y=0.5]{v}
    \Edge[Direct](pl)(u)
    \Edge[Direct](pu)(u)
    \Edge[Direct](u)(su)
    \Edge[Direct](su)(v)
    \Edge[Direct](sl)(v)
    \end{tikzpicture}
    \caption{$N_2[u]$}
    \end{subfigure}
    \quad
    \begin{subfigure}[b]{.15\columnwidth}
    \begin{tikzpicture}
    \NeighborVertex[x=0,y=0]{pl}
    \NeighborVertex[x=0,y=1]{pu}
    \Vertex[x=0.5,y=0.5,label={$u$},position=above,distance=.5mm]{u}
    \Vertex[x=1,y=1]{su}
    \Vertex[x=1,y=0]{sl}
    \Vertex[x=1.5,y=0.5]{v}
    \Edge[Direct](pl)(u)
    \Edge[Direct](pu)(u)
    \Edge[Direct](u)(su)
    \Edge[Direct](su)(v)
    \Edge[Direct](sl)(v)
    \end{tikzpicture}
    \caption{$N^-(u)$}
    \end{subfigure}
    \quad
    \begin{subfigure}[b]{.15\columnwidth}
    \begin{tikzpicture}
    \Vertex[x=0,y=0]{pl}
    \Vertex[x=0,y=1]{pu}
    \NeighborVertex[x=0.5,y=0.5,label={$u$},position=above,distance=.5mm]{u}
    \NeighborVertex[x=1,y=1]{su}
    \Vertex[x=1,y=0]{sl}
    \NeighborVertex[x=1.5,y=0.5]{v}
    \Edge[Direct](pl)(u)
    \Edge[Direct](pu)(u)
    \Edge[Direct](u)(su)
    \Edge[Direct](su)(v)
    \Edge[Direct](sl)(v)
    \end{tikzpicture}
    \caption{$N^+_\infty[u]$}
    \end{subfigure}
    \caption{Examples of neighbourhoods: the elements that belong to the set written below the graph are depicted as white squares.}
    \label{fig:neighbourhood}
\end{figure}
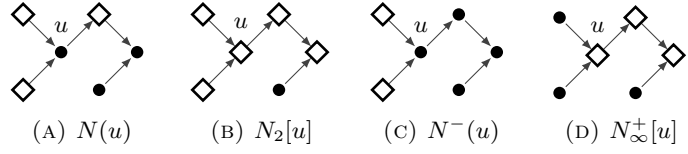
\end{definition}

We now need to restrict the set of graphs we are considering in order to have well-defined advection operators while also retaining graphs that are relevant for applications.

\begin{definition}
We denote by $\mathcal G$ the set of oriented weighted loopless graphs (i.e., no node is directly linked to itself), with upper-bounded degree and lower-bounded edge weight, i.e., for any $G=(V, E, \omega)\in\mathcal G$ there exist $\Delta_G$ and $\delta_G$ such that for any node $u\in V$ and any edge $e\in E$ we have
\[
\deg(u)\le \Delta_G,\quad \omega(e)\ge \delta_G>0.
\]
For any edge $e=(u,v)$ we interpret the weight $\omega(e)$ as the \emph{length} of $e$, i.e., the distance from $u$ to $v$, which we denote by $d_{uv}$.
\end{definition}

\begin{definition}\label{def:function-is-a-sequence}
A function $f$ over a graph $G = (V,E,\omega) \in \mathcal{G}$ is a map $f: V \rightarrow \mathbb{R}$. We denote by 
\[
\| f \|_\infty = \sup_{v \in V} |f(v)|, \text{ and } \| f \|_1 = \sum_{v \in V} |f(v)|,
\]
the supremum and $1$-norm of $f$ respectively, while we denote with $\ell^\infty$ and $\ell^1 \subset \ell^\infty$ the corresponding Banach spaces.
\end{definition}

We emphasize that for graphs in $\mathcal{G}$, the function $f$ can always be represented as an infinite vector in one-to-one correspondence with the nodes of the associated graph $G$. In the particular case where $G$ has finite cardinality, this reduces to a vector in $\mathbb{R}^{|V|}$.

\section{The axiomatic construction}
\label{sec:axiomatic_construction}

The main goal of this paper is to introduce an axiomatic definition of the transport operator on a graph $G \in \mathcal{G}$. This can be described as an operator on~$\ell^\infty$, i.e., the vector space of bounded sequences with the supremum norm (Definition~\ref{def:function-is-a-sequence}).

\begin{definition}
We define a \emph{(discrete) advection operation} as a function
\[
\begin{split}
A : \mathcal{G} &\to B(\ell^\infty)\\
G\in\mathcal{G}&\mapsto A_G,
\end{split}
\]
where $B(\ell^\infty)$ is the set of \emph{bounded linear operators} in $\ell^\infty$.
\end{definition}
An advection operator $A$ maps a graph $G=(V,E,\omega)\in \mathcal{G}$ into a bounded linear operator $A_G:\ell^\infty(V)\to\ell^\infty(V)$, thanks to which we can define the transport equation and its solution~as the following system of first-order ordinary differential equations~(ODE):
\begin{equation}\label{e:ode}
\dfrac{\de }{\de t}f_t(u)=-[A_G\,f_t](u)\text{\quad for any $u\in V$ and $t\ge0$, }
\end{equation}
with initial value $f_0$ and solution
\[
f_t(u)=\bigl[e^{-tA_G}f_0\bigr](u).
\]

Our primary objective is now to derive the conditions that such an $A_G$ must fulfill in order for the solution of~\eqref{e:ode} to exhibit the properties of the discrete analogue of the transport equation. Note that in the case of a finite graph $G$, i.e., when we can identify the function $f$ with a vector $f \in \mathbb{R}^{|V|}$, the ODE~\eqref{e:ode} reads
\[
\begin{split}
\dfrac{\de }{\de t}f_t&=- A_G f_t,\\
f_t&=e^{-tA_G}f_0.
\end{split}
\]

\subsection{Locality}

The first requirement we impose on the advection operator $A$ is a locality constraint, which makes $A$ the analogous of a \emph{first-order} discretization. 

\begin{axiom}[Locality]\label{axiom:locality}
We say that the advection operator $A$ satisfies the axiom of \emph{Locality} if, for any $G\in\mathcal G$ and any node $v$ it holds
\[
[A_G\1_v](u)=
\begin{cases}
\bigl[A_{N[v]}\1_v\bigr](u) & \text{if $u\in N[v]$}\\
0 & \text{otherwise.}
\end{cases}
\]
In the reminder of this paper, we set
\[
a_{uv}=[A_G\1_v](u).
\]
\end{axiom}

This axiom's formulation can be interpreted as follows: consider two graphs $G=(V,E,\omega)$ and $G'=(V',E',\omega')$ and two nodes $v\in V$ and $v'\in V'$ with isomorphic neighborhoods. Then, according to the axiom of Locality, the advection evolution of the functions $\1_v$ and $\1_{v'}$ should be the same, at least to the first order in time, as it holds
\[
e^{-tA_G} \1_v=\1_v-tA_G\1_v+o(t).
\]

\begin{remark}
\label{r:locality}
As we mentioned, the axiom of Locality leads to a first-order discretization, as for any $f\in \ell^\infty$ we have
\[
[A_Gf](u)=\sum_{v\in V}f(v)[A_G\1_v](u)=\sum_{v\in N[u]}a_{uv}f(v),
\]
i.e., $[A_Gf](u)$ only depends on the values of $f$ on the neighborhood of $u$. \end{remark}

The main role of this axiom is to incorporate the topology of the graph into the definition of the advection operator. However, the choice of considering only the first-order neighbors has little to do with our model interpretation, though it leads to a simpler and more concise description. 

\subsection{Advection motion as mass transfer}

The next step is to ensure that the graph advection behaves likes a \emph{transfer of mass}, in analogy with continuous advection, i.e.,
\begin{itemize}
	\item if $f_0$ is a nonnegative function, then the time evolution $f_t=e^{-A_G}f_0$ should remain nonnegative for any $t\ge 0$;
	\item the total ``mass'' (i.e., the sum of the values of $f_t$) should be conserved in the time evolution.
\end{itemize}
We impose the above conditions in the two \emph{Mass Transfer} axioms.

\begin{axiom}[Mass Transfer I]\label{axiom:mass-transfer-I}
We say that the advection operator $A$ satisfies the axiom of \emph{Mass Transfer I} if, given $G\in \mathcal{G}$, we have
\[
f_t=e^{-tA_G}f_0\ge0
\]
for any nonnegative $f_0\in \ell^\infty$ and $t\ge 0$.
\end{axiom}

Since we aim to characterize the entries of the operator, it is more appropriate to reformulate Axiom~\ref{axiom:mass-transfer-I} in terms of $A_G$, rather than its exponential. For graphs $G \in \mathcal{G}$ with a finite number of nodes, such characterization can be obtained as a consequence of the fact that the set of essentially nonnegative matrices, i.e., matrices $P$ such that $[P]_{u,v} \geq 0$ for $u \neq v$, coincides precisely with those for which $e^{t P}$ is nonnegative for all $t \geq 0$; see~\cite[Theorem 10.29]{HighamBook} or the original result by Varga~\cite[Theorem~8.2]{MR1753713}. The next proposition shows that in the case of a (possibly infinite) graph $G \in \mathcal{G}$, the characterization and the proof remain substantially unchanged.

\begin{proposition}
\label{p:varga}
The operator $A$ satisfies the axiom of Mass Transfer~I if and only if, for any $G\in \mathcal{G}$ and $u\ne v$, we have
\begin{equation}
\label{e:positivity}	
a_{uv}=[A_G\1_v](u)\le0.
\end{equation}
\end{proposition}
\begin{proof}
Let $A$ satisfy the axiom of Mass Transfer I and let $f_0\in \ell^\infty$. We have
\[
f_t(u)=\bigl[e^{-tA_G}f_0\bigr](u) = f_0(u)-t[A_Gf_0](u)+o(t)
\]
and choosing $f_0=\1_v$, with $v\ne u$, we obtain
\[
[A_G\1_v](u)=\lim_{t\to 0^+}\dfrac{-f_t(u)}{t}\le 0,
\]
which concludes the first part of the proof.

On the other hand, let us assume that $A$ satisfies~\eqref{e:positivity}. Since $A_G$ is a bounded operator, we have
\begin{equation*}
|a_{uu}|\le\bigl\lVert A_G\1_u\bigr\rVert_\infty\le \lVert A_G\rVert_\infty<\infty.
\end{equation*}
Setting $\alpha_G=\lVert A_G\rVert_\infty$, the (bounded) operator $\tilde A_G=-A_G+\alpha_G\id$, where $\id$ is the identity operator, is such that $[\tilde A_G\1_v](u)\ge0$. It follows that $\tilde A_G f_0\ge0$ for any nonnegative $f_0$, so we have
\[
e^{t\tilde A_G}f_0=
\sum_{k=0}^\infty\dfrac{t^k}{k!}\tilde A_G^kf_0\ge0.
\]
The conclusion follows since $e^{-tA_G} = e^{-\alpha_G} e^{t\tilde A_G}$.
\end{proof}

To formulate an axiom that guarantees the conservation of mass we must move from the setting of $\ell^\infty$ to that of $\ell^1$, that is, consider the subspace of sequences whose series are absolutely convergent. In this way we can ensure that the total mass of the function $f_t$, in the sense of Definition~\ref{def:function-is-a-sequence}, is a well-defined quantity.

\begin{axiom}[Mass Transfer II]\label{axiom:mass-transfer-II}
We say that the advection operator $A$ satisfies the axiom of \emph{Mass Transfer~II} if, given $G\in \mathcal{G}$, for any $f_0\in \ell^1$ and $t\ge 0$, we have
\[
\bigl\lVert f_t\bigr\rVert_1=\bigl\lVert f_0\bigr\rVert_1,
\]
where $f_t=e^{-tA_G}f_0$.
\end{axiom}

Again, we need to express Axiom~\ref{axiom:mass-transfer-II} as a property of the entries of the operator~$A_G$. In the case where $G$ has a finite number of nodes, the characterization reduces to requiring that the sum of the entries of columns of the matrix representing $A_G$ be equal to zero \cite{MR563986} and, again, the same characterization remains valid for the entire class~$\mathcal{G}$.

\begin{proposition}
\label{p:mass_conservation}
Let $A$ satisfy the axiom of Mass Transfer I. Then $A$ satisfies the axiom of Mass Transfer II if and only if, for any $G=(V,E, \omega)\in\mathcal{G}$ and any $v\in V$, we have
\begin{equation}
\label{e:column_sums}
\sum_{u\in V}a_{uv}=\sum_{u\in V}[A_G\1_v](u)=0.
\end{equation}
\end{proposition}

To prove Proposition~\ref{p:mass_conservation} we need the following preliminary result.
\begin{lemma}
\label{l:mass_conservation}
Le $A$ satisfy both~\eqref{e:column_sums} and the axiom of Mass Transfer I. Then for any $G\in\mathcal G$ the operator $A_G$ is bounded in $\ell^1$ with $\lVert A_G\rVert_1\le 2\lVert A_G\rVert_\infty$ and it holds
\[
\sum_{u\in V}[A_G^kf](u)=0
\]
for any $f\in \ell^1$ and any $k\ge 1$.
\end{lemma}

\begin{proof}
Because of the axiom of Mass Transfer I, we have $[A_G\1_v](u)\le 0$ for $u\ne v$ and then~\eqref{e:column_sums} implies that $[A_G\1_v]_v\ge 0$. Hence, for any $f\in \ell^1$ it holds
\begin{equation}
\label{e:mass_conservation_proof_abs_conv}
\begin{split}
\lVert A_G f\rVert_1
&\le\sum_{u\in V}\sum_{v\in V} |f(v)\,a_{uv}|\\
&=\sum_{v\in V} |f(v)|\biggl(a_{vv}-\sum_{u\ne v}a_{uv}\biggr)\\
&=2\sum_{v\in V}|f(v)|\cdot a_{vv}\le 2\lVert A_G\rVert_\infty\lVert f\rVert_1
\\
\end{split}
\end{equation}
and therefore $\lVert A_G\rVert_1\le 2\lVert A_G\rVert_\infty$. As a consequence, for any $f\in \ell^1$ it holds
\[
\sum_{u\in V}[A_Gf](u)=\sum_{u\in V}\sum_{v\in V}f(v)a_{uv}=\sum_{v\in V}f(v)\sum_{u\in V}a_{uv}=0,
\]
where we can change the order of summation because of the absolute convergence proved in~\eqref{e:mass_conservation_proof_abs_conv}. In particular, for any node $v$ and $k\ge1$ we have
\[
\sum_{u\in V}\bigl[A_G^k\1_v\bigr](u)=\sum_{u\in V}\bigl[A_G \underbrace{A_G^{k-1}\1_v}_{\in \ell^1}\,\bigr](u)=0.\qedhere
\]
\end{proof}

\begin{proof}[Proof of Proposition~\ref{p:mass_conservation}]
Let $A$ be an operator that satisfies both the axioms of Mass Transfer.
We rely again on the idea (Varga's trick) from the proof of~\cite[Theorem~8.2]{MR1753713}: setting $\alpha_G=\lVert A_G\rVert_\infty$, we have
\[
e^{\alpha_Gt}\bigl\lVert e^{-tA_G}\1_v\bigr\rVert_1
=\bigl\lVert e^{t(\alpha_G\id-A_G)}\1_v\bigr\rVert_1\\
=\sum_{u\in V}\sum_{k=0}^\infty\dfrac{t^k}{k!}\bigl[(\alpha_G\id-A_G)^k\1_v\bigr](u).
\]
Since every term is nonnegative, the sum converges absolutely and we can change the order of summation. We obtain
\[
e^{\alpha_Gt}\underbrace{\bigl\lVert e^{-tA_G}\1_v\bigr\rVert_1}_{=1}=1+t\left(\alpha_G-\sum_{u\in V}[A_G\1_v](u)\right)+o(t),
\]
which implies~\eqref{e:column_sums}.

On the other hand, let $A$ satisfy the axiom of Mass Transfer I and suppose that~\eqref{e:column_sums} holds. Given a nonnegative $f_0\in \ell^1$, for $t\ge0$ we have
\begin{equation}
\label{e:mass_conservation_proof_a}
\bigl\lVert f_t\bigr\rVert_1 
= \sum_{u\in V} \bigl[e^{-tA_G}f_0\bigr](u)
= \bigl\lVert f_0\bigr\rVert_1 + \sum_{u\in V}\sum_{k=1}^\infty\dfrac{(-t)^k}{k!}\sum_{v\in V}f_0(v)\bigl[A_G^k\1_v\bigr](u).
\end{equation}
Thanks to Lemma~\ref{l:mass_conservation}, we understand that the sum absolutely converges:
\[
\sum_{u\in V}\sum_{k=1}^\infty\sum_{v\in V}\dfrac{t^k}{k!}f_0(v)\left|\bigl[A_G^k\1_v\bigr](u)\right|
=\sum_{k=1}^\infty\dfrac{t^k}{k!}\sum_{v\in V}f_0(v)\left\lVert A_G^k\1_v\right\rVert_1\le \bigl\lVert f_0\bigr\rVert_1e^{t\lVert A_G\rVert_1}
\]
so we can change the order of summation of~\eqref{e:mass_conservation_proof_a} and obtain
\[
\bigl\lVert f_t\bigr\rVert_1
=\bigl\lVert f_0\bigr\rVert_1 + \sum_{k=1}^\infty\dfrac{(-t)^k}{k!}\sum_{v\in V}f_0(v)\underbrace{\sum_{u\in V}\bigl[A_G^k\1_v\bigr](u)}_{=0}=\bigl\lVert f_0\bigr\rVert_1,
\]
where the last equality is due to~Lemma~\ref{l:mass_conservation}.
\end{proof}

Thanks to the first Gershgorin's Theorem \cite{HornJohnson}, we can formulate a stability result for graphs $G \in \mathcal{G}$ of finite size.
\begin{proposition}
\label{p:stability}
If $A$ satisfies both the axioms of Mass Transfer, than for any finite $G\in \mathcal G$, the matrix $-A_G$ is semistable, i.e., $\Re(\lambda)\le0$ for every eigenvalue $\lambda$ of~$-A_G$.
\end{proposition}

\begin{remark}
Under the assumptions of Proposition~\ref{p:stability}, the matrix~$A_G$ is a singular M-matrix, i.e., a matrix whose off-diagonal entries are less than or equal to zero and whose eigenvalues have nonnegative real parts \cite{berman1994nonnegative}.
\end{remark}

We conclude the section by providing some examples of $A_G$ operators on a finite graph that show what matrices look like in cases where neither, one, or both of the Axioms~\ref{axiom:mass-transfer-I} and~\ref{axiom:mass-transfer-II} are satisfied.

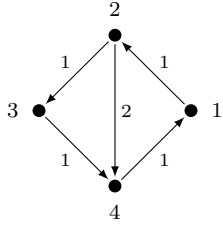
\begin{figure}
\centering
\begin{tikzpicture}
\tikzgraphsettings
\begin{scope}
\foreach \i in {1,...,4} {
	\node (\i) [label={[label distance=2]{(\i-1)*360/4}:$\i$}] at ({(\i-1)*360/4}:1) {};
}
\foreach \i in {1,...,4} {
	\ifnum\i=4 \pgfmathtruncatemacro\next{1}
	\else\pgfmathtruncatemacro\next{\i+1}\fi
	\draw (\i) [->] to (\next);
}
\draw (2) [->] to (4);
\node[empty] at (.65,.65) {\scriptsize$1$};
\node[empty] at (-.65,.65) {\scriptsize$1$};
\node[empty] at (.65,-.65) {\scriptsize$1$};
\node[empty] at (-.65,-.65) {\scriptsize$1$};
\node[empty] at (.15,0) {\scriptsize$2$};
\end{scope}
\end{tikzpicture}
\caption{An example of a finite graph $G\in\mathcal G$. The number on each edge specifies its length.}
\label{f:mass_transfer}
\end{figure}
\begin{example}\label{example:four-advections}
Consider the following advection operators, defined for a graph $G\in\mathcal G$ and a function $f\in \ell^\infty$:
\[
\begin{split}
\bigl[A^{(1)}_Gf\bigr](u)&=\sum_{v\in N^+(u)}\dfrac{f(v)-f(u)}{d_{uv}}\\
\bigl[A^{(2)}_Gf\bigr](u)&=\sum_{v\in N^-(u)}\dfrac{f(u)-f(v)}{d_{vu}}\\ 
\bigl[A^{(3)}_Gf\bigr](u)&=\left[\sum_{w\in N^+(u)}\dfrac{1}{d_{uw}}\right]f(u)-\sum_{v\in N^-(u)}\dfrac{f(v)}{d_{vu}}\\
\bigl[A^{(4)}_Gf\bigr](u)&=\left[\sum_{w\in N^+(u)}\dfrac{1}{\deg^+(u)\,d_{uw}}\right]f(u)-\sum_{v\in N^-(u)}\dfrac{f(v)}{\deg^+(v)\,d_{vu}},
\end{split}
\]
where $A_G^{(1)}$ and $A_G^{(2)}$ correspond, respectively, to a generalization of \emph{forward} and \emph{backward Euler}, while $A_G^{(3)}$ and $A_G^{(4)}$ are modified versions of $A_G^{(2)}$. In the case of the finite graph of Figure~\ref{f:mass_transfer}, the four operators correspond to the following matrices: 
\begin{equation}
\label{e:A1-4}
\begin{aligned}
A^{(1)}_G&=
\text{\small$\begin{bmatrix}
  -1 & 1 & 0 & 0\\
  0 & -3/2 & 1 & 1/2\\
  0 & 0 & -1 & 1\\
  1 & 0 & 0 & -1\\
\end{bmatrix}$}\!,
& A^{(2)}_G&=
\text{\small$\begin{bmatrix}
  1 & 0 & 0 & -1\\
  -1 & 1 & 0 & 0\\
  0 & -1 & 1 & 0\\
  0 & -1/2 & -1 & 3/2\\
\end{bmatrix}$}\!,\\[2ex]
A^{(3)}_G&=
\text{\small$\begin{bmatrix}
  1 & 0 & 0 & -1\\
  -1 & 3/2 & 0 & 0\\
  0 & -1 & 1 & 0\\
  0 & -1/2 & -1 & 1\\
\end{bmatrix}$}\!,
& A^{(4)}_G&=
\text{\small$\begin{bmatrix}
  1 & 0 & 0 & -1\\
  -1 & 3/4 & 0 & 0\\
  0 & -1/2 & 1 & 0\\
  0 & -1/4 & -1 & 1\\
\end{bmatrix}$}\!.
\end{aligned}
\end{equation}
We note that $A^{(1)}$ fails to fulfill any of the Mass Transfer axioms, while $A^{(2)}$ only satisfies the axiom of Mass Transfer I. On the other hand, $A^{(3)}$ and $A^{(4)}$ satisfy both axioms.
\end{example}

\subsection{Moving forward and at the right speed}

The advection operator described so far models a motion based on local neighborhood information, with some global conservation constraints. It is time to specify how the motion's direction and speed are influenced by the edges' orientation and length.

The first axiom of Advection ensures that the motion occurs along the direction of the edges, never going backwards: a mass concentrated at a single node $v$ can only move inside $v$'s cone of successors $N^+_\infty[v]$ (Definition~\ref{def:neighbourhood}). 

\begin{axiom}[Advection I]\label{axiom:advection-I}
We say that $A$ satisfies the axiom of \emph{Advection I} if, given $G\in\mathcal{G}$, the support of $f_t=e^{-tA_G}\1_v$ is concentrated at $N^+_\infty[v]$,
for any node $v$ and any $t\ge 0$.
\end{axiom}
We recast it as a requirement on the operator $A_G$ via the following proposition.
\begin{proposition}
\label{p:forwardness}
An advection operator $A$ satisfies the axiom of Advection I (Axiom~\ref{axiom:advection-I}) if and only if, for any $G\in\mathcal{G}$ and any two nodes $v$ and $u\not\in N^+_\infty[v]$, it holds
\begin{equation}
\label{e:forward_condition}
a_{uv}=[A_G\1_v](u)=0.
\end{equation}
\end{proposition}
\begin{proof}
If $A$ satisfies the axiom of Advection I, then $f_t(u)=\bigl[e^{-tA_G}\1_v\bigr](u)$ is constant (null) in time, hence
\[
\dfrac{\de }{\de t}f_t(u)=-[A_Gf_t](u)=0
\]
and the conclusion follows by taking $t=0$.

On the other hand, if~\eqref{e:forward_condition} holds, then we can inductively prove that $[A_G^k\1_v](u)$ is null as well, for any $u\not\in N^+_\infty[v]$ and for any $k\ge 0$. Indeed, we have
\[
[A_G^k\1_v](u)
=\Biggl[A_G^{k-1}\underbrace{\sum_{w\in N^+_\infty[v]}a_{wv}\1_w\!}_{=A_G\1_v}\,\Biggr](u)
=\!\!\sum_{w\in N^+_\infty[v]}a_{wv}\underbrace{\bigl[A_G^{k-1}\1_w\bigr](u)}_{=0}.
\]
The conclusion follows by considering the Taylor expansion of $f_t(u)$.
\end{proof}
\begin{corollary}
\label{c:forwardness}
If $A$ satisfies the axiom of Advection I (Axiom~\ref{axiom:advection-I}), for any $G\in\mathcal G$ and any function $f_0$ whose domain is concentrated at a set $S$ of nodes, we have
\[
e^{-tA_G}f_0(u)=0,\quad\bigl[A_G^kf_0\bigr](u)=0
\]
for any $u\not\in N^+_\infty[S]=\bigcup_{w\in S}N^+_\infty[w]$ and any $t\ge0$ and $k\ge0$.
\end{corollary}

Let us now focus on the relationship between edge lengths and motion speed. Instead of introducing an external velocity field \cite{chapman2011}, we aim to leverage the network topology and edge lengths to define the transport speed. In the context of the continuous advection equation on the real line, the solution at time $t$ involves translating the initial data by $t$, producing a constant-speed motion. While motion along an oriented graph is generally more complex, we can introduce a similar concept of translation by restricting our focus to oriented trees. Later, in Section~\ref{sec:extension}, we will discuss how to extend this approach to more general classes of graphs.

\begin{definition}
\label{d:signed_distance}
Given an oriented tree $T=(V,E,\omega)$ (Definition~\ref{def:oriented-tree}), there exists a \emph{potential} function $\phi:V\to\mathbb{R}$ such that $d_{vu}=\phi_T(u)-\phi_T(v)$ for any $(v,u)\in E$. Such function~$\phi$ is unique up to a constant, so that we can generalize the edge-length notation and define a \emph{signed distance} for any $u,v\in V$ as
\[
d_{vu}=\phi_T(u)-\phi_T(v).
\]
\end{definition}

\begin{definition}
Let $T=(V,E,\omega)\in\mathcal{G}$ be an oriented tree with a node $v$ and a function $f\in\ell^1$, we define the following quantity:
\[
\bar{d}_v(f)=\sum_{u\in V}d_{vu}f(u),
\]
where $d_{vu}$ is the \emph{signed distance} from $v$ to $u$ (Definition~\ref{d:signed_distance}). In the case where $f$ is nonnegative and $\|f\|_1=1$, we have that $\bar d_v(f)$ represents the \emph{average displacement} of $f$ from $v$.
\end{definition}

In stating the second axiom of Advection, we consider a unit mass concentrated at a node $v$ of an oriented tree. During advective motion, the mass may divide along the branches. However, similarly to translation in the continuous case, we expect that after a time $t$, the average displacement of the mass from the origin $v$ is precisely $t$. Of course, we have to assume that the tree is \emph{leafless}, i.e., every node has at least a successor, to avoid any border effect.
\begin{axiom}[Advection II]\label{axiom:advection-II}
We say that $A$ satisfies the axiom of \emph{Advection II} if, for any leafless oriented tree $T\in\mathcal{G}$ and any node $v$, setting $f_t=e^{-tA_T}\1_v$ we have
\[
\bar{d}_v\bigl(f_t)=t
\]
for any $t\ge 0$.
\end{axiom}
\begin{proposition}
\label{p:advection}
Let $A$ be an advection operator that satisfies the axioms of Mass Transfer (Axioms~\ref{axiom:mass-transfer-I},\ref{axiom:mass-transfer-II}) and Advection I (Axiom~\ref{axiom:advection-I}). Then $A$ satisfies the axiom of Advection II~(Axiom~\ref{axiom:advection-II}) if and only if
\begin{equation}
\label{e:advection_proof}
\sum_{u\ne v}d_{vu}\,a_{uv}=-1
\end{equation}
for any node $v$ of a leafless oriented tree $T\in\mathcal{G}$, where $d_{uv}$ is the signed distance (see Definition~\ref{d:signed_distance}).
\end{proposition}

We will need the following technical lemma.
\begin{lemma}
\label{l:advection}
Let $A$ be an advection operator that satisfies~\eqref{e:advection_proof} and the axioms of Mass Transfer (Axioms~\ref{axiom:mass-transfer-I} and~\ref{axiom:mass-transfer-II}) and Advection I~(Axiom~\ref{axiom:advection-I}). Let $T$ be an oriented tree of $\mathcal G$ and let $\mathcal F$ be the set of $f\in \ell^1$ with support in $N^+_\infty[v]$ and such that $\bar{d}_v\bigl(|f|\bigr)<\infty$. Then, there exists a constant $\beta_T$ such that, for any $f\in \mathcal F$ and for any $k\ge 0$, it holds
\[
\bar d_v\bigl(|A_T^k f|\bigr)\le\beta_T^k\,\Bigl(k\|f\|_1 +\bar d_v\bigl(|f|\bigr)\Bigr).
\]
Moreover, we have
\[
\bar d_v\bigl(A_T^k f\bigr)=
\begin{cases}
\displaystyle-\sum_{w\in N^+_\infty[v]}f(w) & \text{\upshape if $k=1$}\\
0 & \text{\upshape if $k\ge2$.}
\end{cases}
\]
\end{lemma}

\begin{proof}
We have $f(w)=0$ for $w\not\in N^+_\infty[v]$ and, by Corollary~\ref{c:forwardness}, $[A_Tf](u)=0$ for $u\not\in N^+_\infty[v]$, and thus
\[
\begin{split}
\bar d_v\bigl(|A_T f\bigr|)&\le\sum_{u, w\in N^+_\infty[v]}d_{vu}\bigl|a_{uw}f(w)\bigr|\\[-1ex]&=\sum_{w}|f(w)|\Biggl(\!\underbrace{d_{vw}a_{ww}}_{\ge0}-\sum_{u\ne w}\underbrace{d_{vu}a_{uw}}_{\le 0}\!\Biggr).
\end{split}
\]
Because of Proposition~\ref{p:mass_conservation}, we can rewrite the term inside the parentheses as
\[
\begin{split}
-d_{vw}\sum_{u\ne w}a_{uw}-\sum_{u\ne w}d_{vu}a_{uw}
=\sum_{u\ne w}\bigl[\,\underbrace{(d_{vw}-d_{vu})}_{=-d_{wu}}a_{uw}-2d_{vw}a_{uw}\bigr]
\end{split}
\]
which, thanks to Proposition~\ref{p:forwardness} and Equations~\eqref{e:mass_conservation_proof_abs_conv} and~\eqref{e:advection_proof}, is equal to
\[
1+2d_{vw}\sum_{u\ne w}\bigl|a_{uw}\bigr|\le 1+2d_{vw}\lVert A_T\rVert_1,
\]
and therefore
\begin{equation}
\label{e:advection_proof_inner_bound}
\bar d_v\bigl(|A_Tf|\bigr)\le
\lVert f\rVert_1 + 2\lVert A_T\rVert_1\,\bar d_v\bigl(|f|\bigr)\le
\beta_T\Bigl(\|f\|_1 +\bar d_v\bigl(|f|\bigr)\Bigr),
\end{equation}
where $\beta_T=\max\bigl\{1,\, 2\lVert A_T\rVert_1\bigr\}$.
This also implies that the series
\begin{equation*}
\bar d_v\bigl(A_T f\bigr)=
\sum_{u\in N^+_\infty[v]}d_{vu}\sum_{w\in N^+_\infty[v]}a_{uw}f(w).
\end{equation*}
converges absolutely. Therefore we can change the order of summation and obtain
\begin{equation}
\label{e:advection_proof_order_1}
\bar d_v\bigl(A_T f\bigr)=\sum_{w\in N^+_\infty[v]}f(w)\left(\,\sum_{u\in N^+_\infty[v]}(d_{vw}+d_{wu})\,a_{uw}\!\right)
=-\sum_{w\in N^+_\infty[v]}f(w).
\end{equation}
From~\eqref{e:advection_proof_inner_bound} and Corollary~\ref{c:forwardness}, it follows that $A_Tf\in \mathcal F$, hence we can substitute $f$ with $A_Tf$ in~\eqref{e:advection_proof_order_1} and, thanks to Lemma~\ref{l:mass_conservation}, obtain
\[
\bar d_v \bigl(A_T^2f\bigr)=-\sum_{w\in N^+_\infty[v]}[A_Tf](w)
=0.
\]
More generally, for $k\ge 0$ we can prove by induction that $\bar d_v\bigl(|A_T^kf|\bigr)\le\beta_T^k\bigl(k\|f\|_1+
\bar d_v(|f|)\bigr)<\infty$ and thus
\[
\bar d_v\bigl(A_T^{2+k}f \bigr)=\bar d_v\Bigl(A_T^2\bigl(A_T^k f\bigr)\Bigr)=0.\qedhere
\]
\end{proof}

\begin{proof}[Proof of Proposition~\ref{p:advection}]
Let us assume that $A$ satisfies the axioms of Mass Transfer and Advection~I. We rely again on Varga's trick: setting $\alpha_T=\lVert A_T\rVert_\infty$, we write
\[
\begin{split}
te^{\alpha_Tt}&=
\bar{d}_v\left(e^{t(\alpha_T\id-A_T)}\1_v\right)
=\sum_{u\in N^+_\infty[v]}d_{vu}\sum_{k=0}^\infty\,\underbrace{\dfrac{t^k}{k!}\bigl[(\alpha_T\id-A_T)^k\1_v\bigr](u)}_{\ge0}
\\&=-t\cdot\!\!\!\!\!\sum_{u\in N^+_\infty[v]}d_{vu}[A_T\1_v](u)+o(t).
\end{split}
\]
The first order expansion in $t$, together with Proposition~\ref{p:forwardness}, concludes the first part of the proof.

On the other hand, let us assume that $A$ satisfies~\eqref{e:advection_proof}. We can write
\[
\bar d_v\bigl(e^{-tA_T}\1_v\bigr)
=-t\underbrace{\bar d_v\bigl(A_T\1_v\bigr)}_{=-1}\,+\,\,{\bar d_v}\left(\sum_{k=2}^\infty\dfrac{(-t)^k}{k!}A_T^k\1_v\right)\!,
\]
so that it only remains to prove that the last term vanishes. Lemma~\ref{l:advection} lets us prove that it converges absolutely:
\[
\sum_{u\in V}\sum_{k=0}^\infty\left|\dfrac{(-t)^k}{k!}d_{vu}\,\bigl[A_T^k\1_v\bigr](u)\right|
=\sum_{k=0}^\infty\dfrac{t^k}{k!}\bar d_v\bigl(|A_T^k\1_v|\bigr)<\infty
\]
and we obtain
\[
\bar d_v\left(\sum_{k=2}^\infty\dfrac{(-t)^k}{k!}A_T^k\1_v\right)=\sum_{k=2}^\infty\dfrac{(-t)^k}{k!} \underbrace{\bar d_v\bigl(A_T^k\1_v\bigr)}_{=0}=0,
\]
which concludes the proof.
\end{proof}

\begin{remark}
\label{r:extension_coherent_distances}
Reducing ourselves to the case of trees without leaves may seem reductive, however the class of graphs in $\mathcal{G}$ for which this result makes sense is broader. It can be proved that if~\eqref{e:advection_proof} holds, then Advection II is indeed satisfied for all oriented graphs in $\mathcal G$ such that every node has at least a successor and for which there exists a potential $\phi$ like in Definition~\ref{d:signed_distance}.
\end{remark}

We can now go back to analyzing the proposals of advection operators that we observed in Example~\ref{example:four-advections}, in particular we can assess which of the new axioms introduced here are compatible with the proposed realization of the advection operator.

\begin{figure}
\begin{subfigure}[b]{.47\textwidth}
\centering
\begin{tikzpicture}
\tikzgraphsettings
\begin{scope}[xshift=-1cm]
\def\n{5}
\foreach \i in {1,...,\n} {
	\ifnum\i=1
		\node (\i) at ({(\i-1)*360/\n}:1) {};
	\else
		\node (\i) [label={[label distance=-1]{(\i-1)*360/\n}:$\number\numexpr\i-\n-1$}] at ({(\i-1)*360/\n}:1) {};
	\fi
}
\foreach \i in {1,...,\n} {
	\ifnum\i=\n \pgfmathtruncatemacro\next{1}
	\else\pgfmathtruncatemacro\next{\i+1}\fi
	\draw (\i) [->] to (\next);
}
\end{scope}
\begin{scope}[xshift=1cm, xscale=-1]
\def\n{9}
\foreach \i in {1,...,\n} {
	\node (-\i) [label={[label distance=2]{180-(\i-1)*360/\n}:$\number\numexpr\i-1$}]  at ({(\i-1)*360/\n}:1) {};
}
\foreach \i in {1,...,\n} {
	\ifnum\i=\n \pgfmathtruncatemacro\next{-1}
	\else\pgfmathtruncatemacro\next{-\i-1}\fi
	\draw (-\i) [->] to (\next);
}
\end{scope}
\end{tikzpicture}
\caption{A graph of $\mathcal G$, composed of two oriented cycles, with a node in common. The left cycle has edges of length $1/5$, while the edges of the right one are of length $1/9$.}
\label{f:advection}
\end{subfigure}\hfill
\begin{subfigure}[b]{.47\textwidth}
\centering\footnotesize
\begin{tikzpicture}
\begin{axis}[
    xlabel={$u$},
    ylabel={},
	xtick = {-4,-2,0,2,4,6,8},
    ytick = {0,0.07692307692306885,0.15},
    width=\textwidth,
    height=2.7cm,
    ymin=0,
    ymax=0.16,
    yticklabel style={
        /pgf/number format/fixed,
    },
]
\addplot [color=black, mark=*] coordinates {
(-4, 0.07692307692306885)
(-3, 0.07692307692306886)
(-2, 0.07692307692306885)
(-1, 0.07692307692306885)
(0, 0.07692307692306885)
(1, 0.07692307692306882)
(2, 0.07692307692306885)
(3, 0.07692307692306884)
(4, 0.07692307692306884)
(5, 0.07692307692306882)
(6, 0.07692307692306882)
(7, 0.07692307692306882)
(8, 0.07692307692306882)
};
\end{axis}
\end{tikzpicture}
\begin{tikzpicture}
\begin{axis}[
    xlabel={$u$},
	xtick = {-4,-2,0,2,4,6,8},
    ytick = {0,0.07,0.15},
    width=\textwidth,
    height=2.7cm,
    ymin=0,
    ymax=0.16,
    yticklabel style={
        /pgf/number format/fixed,
    },
]
\addplot [color=black, mark=*] coordinates {
(-4, 0.07142857142856822)
(-3, 0.07142857142856823)
(-2, 0.07142857142856822)
(-1, 0.07142857142856825)
(0, 0.14285714285713644)
(1, 0.07142857142856822)
(2, 0.0714285714285682)
(3, 0.0714285714285682)
(4, 0.07142857142856822)
(5, 0.07142857142856822)
(6, 0.0714285714285682)
(7, 0.0714285714285682)
(8, 0.07142857142856819)
};
\end{axis}
\end{tikzpicture}
\caption{Plot of $f_t(u)$ for $t=100$ with initial unitary mass at node $0$, according to $A^{(3)}$ (above) and $A^{(4)}$ (below).}
\label{f:advection_resonance}
\end{subfigure}
\caption{Long-term comparison of the advection process on the graph in~(\subref{f:advection}), according to operators $A^{(3)}$ and $A^{(4)}$. The latter, which satisfies Advection II, exhibits resonance behavior.}
\end{figure}
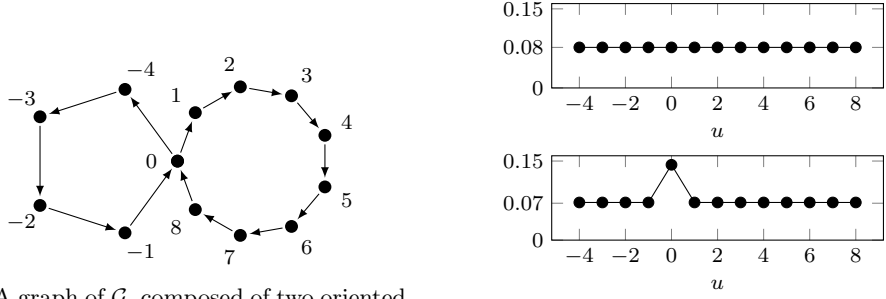

\begin{example}[continued from  Example~\ref{example:four-advections}]
It is straightforward to verify that both the operators $A^{(3)}$ and $A^{(4)}$ of~\eqref{e:A1-4} satisfy the axiom of Advection~I, while only $A^{(4)}$ also satisfies Advection~II.

Let us apply those operators to the graph depicted in Figure~\ref{f:advection}, which comprises two connected cycles, with an initial unitary mass located at the common node $0$. Since both cycles have a total length of $1$, we anticipate that the mass will peak at node $0$ whenever $t$ is an integer. However, for large $t$, we only observe this behavior with $A^{(4)}$, which satisfies the axiom of Advection II. With $A^{(3)}$, on the other hand, the mass flows along the two cycles without perfect synchronization, ultimately converging to a uniform distribution; see Figure~\ref{f:advection_resonance}.
\end{example}

\subsection{Flow through a node}
\label{s:flow}

The axioms introduced so far are sufficient to prove an explicit relation between the \emph{time integration} of the quantity of mass on a node~$u$ up to time $t$ and the \emph{spatial integration} (i.e., the sum) of the mass on the children of $u$ at time $t$. This result, stated in Proposition~\ref{p:lemma_integral}, will be useful for characterizing the next (and final) axiom, but it is insightful in its own right.

\begin{proposition}
\label{p:lemma_integral}
Let $A$ be an advection operator that satisfies the axioms of Locality (Axiom~\ref{axiom:locality}), Mass Transfer I (Axiom~\ref{axiom:mass-transfer-I}) and II (Axiom~\ref{axiom:mass-transfer-II}) and Advection I (Axiom~\ref{axiom:advection-I}.) Let $T\in\mathcal G$ be an oriented tree and let $f_0\in \ell^1$ be a nonnegative function whose domain is concentrated at $N^-_\infty[u]$ for a given node~$u$. Then for any $t\ge0$, setting $f_t=e^{-tA_T}f_0$ we have
\[
a_{uu}\int_0^t f_s(u)\,\de s=\sum_{w\in N^+_\infty[u]\setminus\{u\}} f_t(w).
\]
\end{proposition}

To obtain the result we need the following lemma.
\begin{lemma}
\label{l:lemma_sum}
Let $A$ and $\,T$ be as in Proposition~\ref{p:lemma_integral} and let $g\in \ell^1$ be a nonnegative function whose domain is concentrated at $N^+_\infty\bigl[N^-_\infty[u]\bigr]$, for a given node $u$. Then for any $k>0$ we have
\begin{equation}
\label{e:lemma_sum}
a_{uu}g(u) = -\sum_{i=1}^k\sum_{w\in N^+_i(u)}[A_Tg](w)+\sum_{w\in N_k^+(u)}a_{ww}g(w).
\end{equation}
\end{lemma}

\begin{proof}
We proceed by induction on $k$, and assume that~\eqref{e:lemma_sum} is true for $k\ge0$ (the base case $k=0$ is trivial). For any two nodes $w\in N^+_k(u)$ and $v\in N^+(w)$, by Locality  and Advection I we have:
\[
[A_Tg](v)=a_{vv}g(v)+a_{vw}g(w)+\sum_{r\in N^-(v)\setminus\{w\}}a_{vr}\underbrace{g(r)}_{=0}.
\]
Summing over $v\in N^+(w)$ and then over $w\in N^+_k(u)$, thanks to Proposition~\ref{p:mass_conservation} we obtain
\[
\sum_{v\in N^+_{k+1}(u)}[A_Tg](v)=\sum_{v\in N^+_{k+1}(u)}a_{vv}g(v)+\sum_{w\in N_k^+(u)}\underbrace{\left(\sum_{v\in N^+(w)}a_{vw}\right)}_{=-a_{ww}}g(w).
\]
We now substitute the last term in~\eqref{e:lemma_sum}, resulting in
\[
a_{uu}g(u) = -\sum_{i=1}^{k+1}\sum_{v\in N^+_i(u)}[A_Tg](v)+\sum_{v\in N_{k+1}^+(u)}a_{vv}g(v),
\]
which proves the statement for $k+1$.
\end{proof}

\begin{proof}[Proof of Proposition~\ref{p:lemma_integral}]
For any $s>0$, the domain of the function $f_s$ is concentrated at $N^+_\infty\bigl[N^-_\infty[u]\bigr]$, thanks to the axiom of Advection I (Corollary~\ref{c:forwardness}). Moreover, thanks to Mass Transfer, $f_s$ is a nonnegative $\ell^1$ function and thus we can apply Lemma~\ref{l:lemma_sum} with $g=f_s$: integrating both sides of~\eqref{e:lemma_sum} we get
\[
a_{uu}\int_0^t f_s(u)\,\de s = \sum_{i=1}^k\sum_{w\in N^+_i(u)}\int_0^t\underbrace{-A_tf_s(w)}_{=\frac{\de }{\de s}f_s(w)}\de s+\sum_{w\in N_k^+(u)}a_{ww}\int_0^tf_s(w)\,\de s.
\]
We can prove that the last term goes to $0$ as $k$ tends approaches $\infty$ by showing that the following series is convergent for any $t>0$:
\[
\sum_{k=0}^\infty\sum_{w\in N_k^+(u)}a_{ww}\int_0^tf_s(w)\,\de s
\le\int_0^t\lVert A_T\rVert_\infty\underbrace{\sum_{w\in N^+_\infty[u]}f_s(w)}_{\le\lVert f_s\rVert_1=\lVert f_0\rVert_1}\de s<\infty.
\]
(We could exchange the sum and the integral thanks to the positivity of $f_s$.)

Therefore, for $k\to\infty$ we have
\[
a_{uu}\int_0^tf_s(u)\,\de s = \lim_{k\to \infty}\sum_{i=1}^k\sum_{w\in N^+_i(u)}f_t(w)=\sum_{w\in N^+_\infty[u]\setminus\{u\}} f_t(w).\qedhere
\]
\end{proof}

\subsection{How does the mass split?}

In order to fully characterize the advection operator, we must clarify how the mass splits whenever it encounters a branching. For simplicity, we can continue to focus on oriented trees.

Different intuitions about how the mass splits are feasible, depending on the phenomenon we are modeling. For instance, one may expect that the mass flowing through a node $v$ splits evenly among $v$'s children. Alternatively, the splitting may depend on the distance from $v$ to each of his children. More generally, there might be a priori knowledge on the edges of the graph that suggests a different splitting strategy. For example, in road networks, the splitting may favor larger roads.

Our approach links the speed of mass transfer to the length of the edges. As an example, consider an advection process on oriented tree $T$, with initial mass concentrated at a node $v$. Assume that $v$ is connected to two leaves $u$ and $w$ such that  $d_{vu}=2d_{vw}$, as shown in Figure~\ref{f:simple_tree}. After a time $t$, in a constant-speed motion we expect the mass that has reached $u$ to be half of the mass that has reached $w$, since it has to travel twice the distance. For $t\to\infty$, we anticipate that the mass will have completely left $v$, distributing among $u$ and $w$ in ratio of $1:2$ (Figure~\ref{f:splitting_example}).

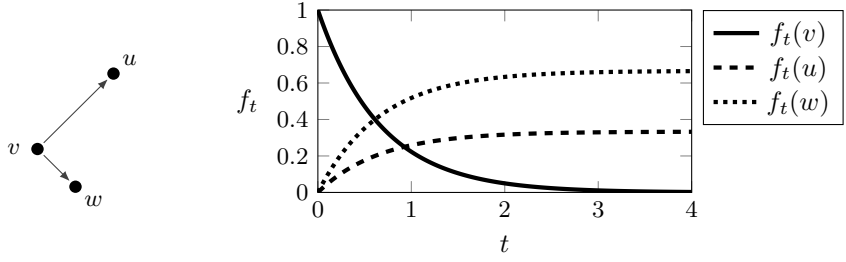
\begin{figure}
\centering
\begin{subfigure}[b]{.19\columnwidth}
\centering
\begin{tikzpicture}
\Vertex[x=0,y=0, label={$v$}, color=mplblue, position=left, distance=0pt]{v}
\Vertex[x=1,y=1, label={$u$}, color=mplorange, position=north east, distance=-3pt]{u}
\Vertex[x=.5,y=-.5, label={$w$}, color=mplgreen, position=south east, distance=-3pt]{w}
\Edge[Direct](v)(u)
\Edge[Direct](v)(w)
\end{tikzpicture}
\vspace{4ex}
\subcaption{A tree where $d_{vu}=2d_{vw}$.}
\label{f:simple_tree}
\end{subfigure}
\quad
\begin{subfigure}[b]{.65\columnwidth}
\centering
\begin{tikzpicture}
\begin{axis}[
    xlabel={$t$},
    ylabel={$f_t$},
    ylabel style = {rotate = -90},
    width=6.5cm,
    height=4cm,
    ymin=0,
    ymax=1,
    xmin=0,
    xmax=4,
    xticklabel style={
        font=\small,
    },
    yticklabel style={
        font=\small,
        /pgf/number format/fixed,
    },
    legend pos = outer north east,
]
\addplot [ultra thick, color=mplblue] table {figures/splitting_v.txt};
\addplot [ultra thick, dashed, color=mplorange] table {figures/splitting_u.txt};
\addplot [ultra thick, dotted, color=mplgreen] table {figures/splitting_w.txt};
\addlegendentry{$f_t(v)$}
\addlegendentry{$f_t(u)$}
\addlegendentry{$f_t(w)$}
\end{axis}
\end{tikzpicture}
\subcaption{Time evolution of $f_t(v)$, $f_t(u)$ and $f_t(w)$.\\\vphantom{M}}
\label{f:splitting_example}
\end{subfigure}
\caption{Example of the advection process on a simple tree, with initial unit mass on~$v$, according to the operator $A^{(4)}$.}
\label{f:spitting_example_on_simple_tree}
\end{figure}

Similarly, in the case of a general tree $T\in\mathcal{G}$, we expect that the total mass in each subtree of $v$'s children will be inversely proportional to the length of the corresponding edge from $v$, after an infinite amount of time. Therefore, we establish the following axiom.

\begin{axiom}[Splitting]\label{axiom:splitting}
We say that $A$ satisfies the axiom of \emph{Splitting} if, given an oriented tree $T\in\mathcal{G}$ and a node $v$, setting $f_t=e^{-tA_T}\1_v$ it holds
\[
\lim_{t\to\infty}\dfrac{\displaystyle\sum_{z\in N^+_\infty[u]}f_t(z)}{\displaystyle\sum_{z\in N^+_\infty[w]}f_t(z)}
=
\dfrac{d_{vw}}{d_{vu}}
\]
for any $u,w\in N^+(v)$.
\end{axiom}

Thanks to the results in Section~\ref{s:flow}, we can compute the total mass lying on each branch of~$v$ after an infinite amount of time in terms of the coefficients of the column relative to $v$ in the advection operator.

\begin{proposition}
\label{p:splitting}
Let $A$ be an advection operator that satisfies the axioms of Locality (Axiom~\ref{axiom:locality}), Mass Transfer I and II (Axioms~\ref{axiom:mass-transfer-I}, \ref{axiom:mass-transfer-II}) and Advection~I (Axiom~\ref{axiom:advection-I}), and let $T\in\mathcal G$ be an oriented tree. Then, given a node $v$, for any $u\in N^+(v)$ and for $f_t=e^{-tA_T}\1_v$, we have
\[
\lim_{t\to\infty}\sum_{w\in N^+_\infty[u]}f_t(w)=-\dfrac{a_{uv}}{a_{vv}}.
\]
\end{proposition}
\begin{proof}
By Locality and Advection I (Corollary~\ref{c:forwardness}), for any $k\ge0$ we have
\[
\begin{cases}
\bigl[A_T^{k+1}\1_v\bigr](v)=\bigl[A_T\bigl(A_T^k\1_v\bigr)\bigr](v) = a_{vv}\bigl[A_T^k\1_v\bigr](v)\\
\bigl[A_T^{k+1}\1_v\bigr](u)=a_{uv}\bigl[A_T^k\1_v\bigr](v)+a_{uu}\bigl[A_T^k\1_v\bigr](u)
\end{cases}
\]
so it can be easily proved by induction that it holds
\[
\bigl[A_T^k\1_v\bigr](v)=a_{vv}^k,\quad
\bigl[A_T^k\1_v\bigr](u)=
\begin{cases}
a_{uv}\cdot k\,a_{vv}^{k-1} & \text{if $a_{uu}=a_{vv}$}\\[1.5ex]
a_{uv}\cdot\dfrac{a_{vv}^k-a_{uu}^k}{a_{vv}-a_{uu}} & \text{otherwise}
\end{cases}
\]
When $a_{uu}\ne a_{vv}$, we have:
\begin{equation}
\label{e:exact_formula}
f_s(u)=\sum_{k=0}^\infty\dfrac{(-s)^k}{k!}\bigl[A_T
^k\1_v\bigr](u)=\dfrac{a_{uv}}{a_{vv}-a_{uu}}\bigl(e^{-sa_{vv}}-e^{-sa_{uu}}\bigr)
\end{equation}
and therefore, from Proposition~\ref{p:lemma_integral}:
\[
\sum_{w\in N^+_\infty[u]} f_t(w)
=f_t(u)+a_{uu}\int_0^t f_s(u)\,\de s
=a_{uv}\left(\dfrac{e^{-ta_{vv}}-e^{-ta_{uu}}}{a_{vv}-a_{uu}}-\dfrac{1}{a_{vv}}\right).
\]
The conclusion follow by taking the limit for $t\to\infty$. The other case, where $a_{uu}=a_{vv}$ is similar.
\end{proof}

The following characterization follows immediately.
\begin{corollary}
\label{c:splitting}
Under the same assumptions of Proposition~\ref{p:splitting}, the advection operator $A$ satisfies the axiom of Splitting (Axiom~\ref{axiom:splitting}) if and only if
\begin{equation}
\label{e:csplitting}
\dfrac{a_{uv}}{a_{wv}}=\dfrac{d_{vw}}{d_{vu}}
\end{equation}
for any node $v$ and any $u,w\in N^+(v)$.
\end{corollary}

\begin{example}[continued from Example~\ref{example:four-advections}]
\label{e:one_more}
Thanks to Corollary~\ref{c:splitting}, we know that both $A^{(3)}$ and $A^{(4)}$ satisfy the axiom of Splitting. Here, we introduce an example of an advection operator that satisfies all the previous axioms but follows a different splitting strategy. Specifically, we impose that the mass flowing through a node splits evenly among its children. To achieve this, guided by Proposition~\ref{p:splitting}, we set
\[
\bigl[A_G^{(5)}f\bigr](u)=\deg^+(u)\,s_u\,f(u)-\sum_{v\in N^-(u)}s_vf(v)\\[-2ex]
\]
where
\[
s_v=
\begin{cases}
\dfrac1{\displaystyle\sum_{w\in N^+(v)}d_{vw}} & \text{if $N^+(v)\ne\varnothing$}\\[6ex]
0& \text{otherwise.}
\end{cases}
\]
The matrix corresponding to $A^{(5)}$ for the graph of Figure~\ref{f:mass_transfer} is
\begin{equation}
\label{e:A5-6}
A^{(5)}_G=
\text{\small$\begin{bmatrix}
  1 & 0 & 0 & -1\\
  -1 & 2/3 & 0 & 0\\
  0 & \phantom{-}-1/3\phantom{-} & 1 & 0\\
  0 & -1/3 & -1 & 1\\
\end{bmatrix}$}\!.
\end{equation}
See Section~\ref{ss:splitting} for examples of advection on trees that illustrate the role of the axiom of Splitting.
\end{example}

As a conclusion of this section, Table~\ref{t:fulfillment} summarizes the axioms satisfied by the advection operators discussed. We have excluded $A^{(1)}$ because it fails to meet the most basic requirements, rendering nearly all propositions inapplicable.
\begin{table}
\def\Y{\textsc{y}}
\def\N{\textsc{n}}
\def\Ys{\textsc{y}\rlap{*}}
\caption{Fulfillment of axioms by advection operators.}\vspace{-.5\baselineskip}
\label{t:fulfillment}
\renewcommand\arraystretch{1.5}
\setlength\belowrulesep{0pt}
\centering
\begin{tabular}{l@{\qquad}*6c}
\toprule
Axiom & $A^{(2)}$ & $A^{(3)}$ & $A^{(4)}$ & $A^{(5)}$ \\
\midrule
Locality & \Y & \Y & \Y & \Y \\
Mass Transfer I  & \Y & \Y & \Y & \Y\\
Mass Transfer II & \N & \Y & \Y & \Y\\
Advection I & \Y & \Y & \Y & \Y\\
Advection II & \N & \N & \Y & \Y\\
Splitting & \N & \Y & \Y & \N\\
\bottomrule
\end{tabular}
\end{table}

\section{Characterization of the advection operator}
\label{sec:characterization}

\subsection{Existence and uniqueness} 

As promised at the beginning of the discussion, the axioms introduced in Section~\ref{sec:axiomatic_construction} are sufficient to uniquely characterize an advection operator on the graphs on $\mathcal G$. 

\begin{theorem}
\label{t:characterization}
There exists a unique advection operator $A\,:\mathcal{G}\to B(\ell^\infty)$ satisfying the axioms of Locality (Axiom~\ref{axiom:locality}), Mass Transfer (Axioms~\ref{axiom:mass-transfer-I}, \ref{axiom:mass-transfer-II}), Advection (Axioms~\ref{axiom:advection-I}, \ref{axiom:advection-II}) and Splitting (Axiom~\ref{axiom:splitting}).
\end{theorem}
\begin{proof}
Consider a graph $G\in\mathcal G$ and a node $v$ in $G$. The axioms of Locality~(Axiom~\ref{axiom:locality}), Mass Transfer I (Axiom~\ref{axiom:mass-transfer-I}) and II~(Axiom~\ref{axiom:mass-transfer-II}) and Advection I (Axiom~\ref{axiom:advection-I}) imply the following conditions:
\[
\begin{aligned}
a_{uv}&=0\quad \text{if $u\not\in N^+[v]$}\\
a_{vv}&=-\displaystyle\sum_{u\in N^+(v)}a_{uv}.\\
\end{aligned}
\]
If $N^+(v)=\varnothing$, then $a_{uv}=0$ for any $u\in V$. Otherwise, in order to apply the remaining axioms, we can construct a leafless oriented tree $T\in\mathcal G$ with a node~$v'$ such that the neighborhoods $N_G[v]$ and $N_T[v']$ are isomorphic. The axiom of Locality (Axiom~\ref{axiom:locality}) implies that, for any node $u$ in $N_G[v]$ and its corresponding node $u'$ in $N_T[v']$, we have:
\[
a_{uv}=[A_G\1_v](u)=[A_T\1_{v'}](u'),
\]
so the coefficients $a_{uv}$ of $A_G$ are the same as the coefficients of $A_T$, for any $u\in V$.  
Therefore, the axiom of Advection II (Axiom~\ref{axiom:advection-II}) implies that
\[
\sum_{u\in N^+(v)}d_{vu}a_{uv}=-1.
\]
Finally, from the axiom of Splitting (Axiom~\ref{axiom:splitting}) it follows that the quantity $d_{vu}a_{uv}$ is the same for any neighbor $u\in N^+(v)$. We can conclude that it holds
\[
a_{uv}=
\begin{cases}
-\dfrac{1}{\deg^+(v)\,d_{vu}} & \text{if $u\in N^+(v)$}\\[2.5ex]
\displaystyle\sum_{w\in N^+(v)}\dfrac{1}{\deg^+(v)\,d_{vw}} & \text{if $u=v$}\\[2.5ex]
0 & \text{otherwise},
\end{cases}
\]
which uniquely characterizes $A_G$. This corresponds to the operator $A^{(4)}$, which we defined in~\eqref{e:A1-4}.

Finally, for any $G\in\mathcal G$ the operator $A_G:\ell^\infty\to \ell^\infty$ defined above is bounded since, for any $f\in \ell^\infty$, we have
\begin{equation}
\label{e:bounded}
\|A_Gf\|_\infty=2\sup_{u\in V}|a_{uu}|\cdot\|f\|_\infty
\le\dfrac{2}{\delta_G} \|f\|_\infty,
\end{equation}
where $\delta_G$ represents a lower bound on the edge lengths.
Propositions~\ref{p:varga}, \ref{p:mass_conservation}, \ref{p:forwardness}, \ref{p:advection} and Corollary~\ref{c:splitting} ensure that $A$ satisfies all the mentioned axioms.
\end{proof}

\subsection{Extension to more general graphs}
\label{sec:extension}

The analysis so far has been restricted to \emph{oriented} graphs, but all the operators defined in~\eqref{e:A1-4} can also be applied to \emph{directed graphs} (where bidirectional edges are allowed) without modifications.

It is easily verified that the axioms of Locality (Axiom \ref{axiom:locality}), Mass Transfer (Axioms~\ref{axiom:mass-transfer-I} and~\ref{axiom:mass-transfer-II}) and of Advection I (Axiom \ref{axiom:advection-I}) are satisfied by both $A^{(3)}$ and $A^{(4)}$ in the context of directed graphs, with the same proofs given in the oriented case.

\begin{remark}
The lack of bidirectional edges in the graphs of $\mathcal{G}$ was actually necessary for the uniqueness of the advection operator proved in Theorem~\ref{t:characterization}. Indeed, given a node $v$ of a graph $G\in\mathcal{G}$, we assumed the existence of an oriented tree with a node $v'$ such that $N[v]$ is isomorphic to $N[v']$, in order to apply the axioms of Advection II and Splitting. This is only true if $G$ is oriented.
\end{remark}

\begin{remark}
\label{r:laplacian}
When $G$ is a simple graph (i.e., an undirected graph with all edge lengths equal to $1$), we have
\[
\begin{split}
[A^{(3)}_Gf](u)&=\deg(u)f(u)-\sum_{v\in N(u)}f(v)\\
[A^{(4)}_Gf](u)&=f(u)-\sum_{v\in N(u)}\dfrac{f(v)}{\deg(v)},
\end{split}
\]
which correspond to the combinatorial Laplacian and to the right-normalized Laplacian, respectively.
\end{remark}

\begin{remark}
The bound in~\eqref{e:bounded} ensures that $A^{(4)}_G:\ell^\infty\to\ell^\infty$ is a bounded operator without any assumption on the maximum degree of the nodes of $G$. Therefore, in the special case of $A^{(4)}$, we can further extend the class $\mathcal G$ by by permitting unbounded, yet finite, degrees.
\end{remark}

\section{Analytical and numerical examples}
\label{sec:advection_dynamics}

The relevant Python code used for the experiments and analyses presented in the following sections is available as a GitHub repository \texttt{\href{https://github.com/francesco-zigliotto/graph-advection/blob/main/graph_advection.ipynb}{francesco-zigliotto/graph- advection}}.

\subsection{Motion on an infinite rectangular grid}

To demonstrate the dynamics induced by the operators we have defined in Examples~\ref{example:four-advections} and~\ref{e:one_more}, let us consider some illustrative examples.

\begin{figure}
\centering
\begin{subfigure}[b]{0.35\columnwidth}
\begin{tikzpicture}[xscale=1.22, yscale=.7,baseline, anchor=base]
\SetEdgeStyle[LineWidth=.4pt]
\Vertex[x=0.5,y=0,color=white,size=0]{A}
\Vertex[x=1,y=0,color=white,size=0]{B}
\Vertex[x=2,y=0,color=white,size=0]{C}
\Vertex[x=3,y=0,color=white,size=0]{D}
\Vertex[x=3.5,y=0,color=white,size=0]{E}
\Vertex[x=0.5,y=1,color=white,size=0]{F}
\Vertex[x=1,y=1]{G}
\Vertex[x=2,y=1]{H}
\Vertex[x=3,y=1]{I}
\Vertex[x=3.5,y=1,color=white,size=0]{J}
\Vertex[x=0.5,y=2,color=white,size=0]{K}
\Vertex[x=1,y=2]{L}
\Vertex[x=2,y=2]{M}
\Vertex[x=3,y=2]{N}
\Vertex[x=3.5,y=2,color=white,size=0]{O}
\Vertex[x=0.5,y=3,color=white,size=0]{P}
\Vertex[x=1,y=3]{Q}
\Vertex[x=2,y=3]{R}
\Vertex[x=3,y=3]{S}
\Vertex[x=3.5,y=3,color=white,size=0]{T}
\Vertex[x=0.5,y=4,color=white,size=0]{U}
\Vertex[x=1,y=4,color=white,size=0]{V}
\Vertex[x=2,y=4,color=white,size=0]{W}
\Vertex[x=3,y=4,color=white,size=0]{X}
\Vertex[x=3.5,y=4,color=white,size=0]{Y}
\Edge[Direct,style={dashed}](B)(G)
\Edge[Direct,style={dashed}](C)(H)
\Edge[Direct,style={dashed}](D)(I)
\Edge[Direct,style={dashed}](F)(G)
\Edge[Direct,style={dashed}](K)(L)
\Edge[Direct,style={dashed}](I)(J)
\Edge[Direct,style={dashed}](N)(O)
\Edge[Direct,style={dashed}](P)(Q)
\Edge[Direct,style={dashed}](Q)(V)
\Edge[Direct,style={dashed}](R)(W)
\Edge[Direct,style={dashed}](S)(X)
\Edge[Direct,style={dashed}](S)(T)
\Edge[Direct](G)(H)
\Edge[Direct](H)(I)
\Edge[Direct](L)(M)
\Edge[Direct](M)(N)
\Edge[Direct](G)(L)
\Edge[Direct](H)(M)
\Edge[Direct](I)(N)
\Edge[Direct](L)(Q)
\Edge[Direct](M)(R)
\Edge[Direct](N)(S)
\Edge[Direct](Q)(R)
\Edge[Direct](R)(S)
\draw[|-|] (K.west)--(P.west) node[midway, above, sloped] {$d_y$};
\draw[|-|] (V.north)--(W.north) node[midway, above, sloped] {$d_x$};
\end{tikzpicture}
\caption{The infinite oriented rectangular grid $G$, with a unit cell of length $d_x$ and height $d_y$. In the simulation, $d_x=3$ and $d_y=1$.}
\label{fig:infinite-grid}
\end{subfigure}\hspace{1em}
\begin{subfigure}[b]{.6\columnwidth}
\centering
\begin{tikzpicture}
\begin{axis}[
    xmin=0, xmax=120,
    ymin=0, ymax=120,
    xlabel={$x$},
    ylabel={$y$},
    width=7cm,
    height=7cm,
    ]
    \addplot graphics[xmin=0, xmax=120, ymin=0, ymax=120] {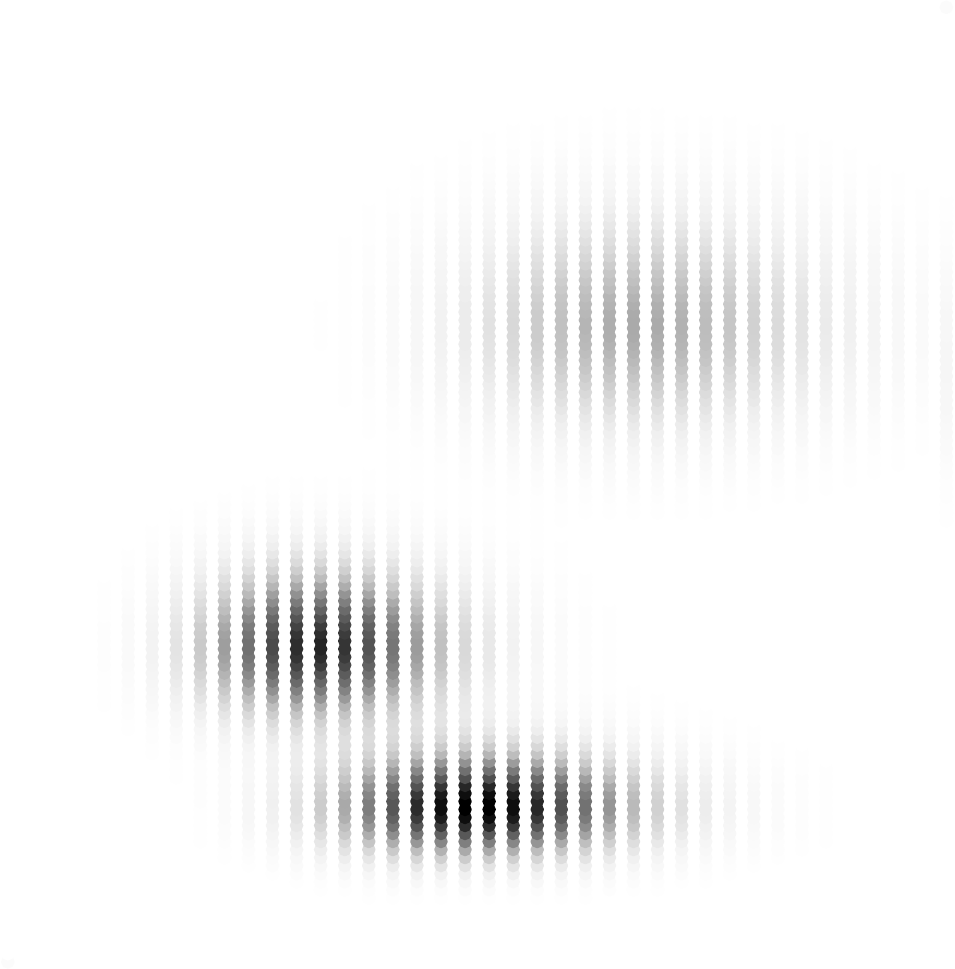};
    \draw[] (axis cs:80,80) -- (axis cs:90,100);
    \node[anchor=south] at (axis cs:90,100) {$A^{(3)}$};
    \draw[] (axis cs:40,40) -- (axis cs:20,55);
    \node[anchor=south] at (axis cs:20,55) {$A^{(4)}$};
    \draw[] (axis cs:60,20) -- (axis cs:90,40);
    \node[anchor=south west] at (axis cs:90,40) {$A^{(5)}$};
    \draw[dashed] (axis cs:40,40) -- (axis cs:40,0);
    \draw[dashed] (axis cs:40,40) -- (axis cs:0,40);
    \draw[dashed] (axis cs:80,80) -- (axis cs:80,0);
    \draw[dashed] (axis cs:80,80) -- (axis cs:0,80);
    \draw[dashed] (axis cs:60,20) -- (axis cs:60,0);
    \draw[dashed] (axis cs:60,20) -- (axis cs:0,20); 
\end{axis}
\end{tikzpicture}
\caption{Mass concentration in the $x,y$ plane at time $t = 80$, starting with a unit mass at the bottom-left corner, according to operators $A^{(3)}$, $A^{(4)}$, and $A^{(5)}$. The expected value of the mass distribution is highlighted.}
\label{f:grid_simulation}
\end{subfigure}
\caption{Depiction of the infinite grid and simulation on a $40\times 120$ grid of the mass distribution according to three different advection operators. The simulation was conducted on a sufficiently large truncated grid to minimize border effects.}
\end{figure}
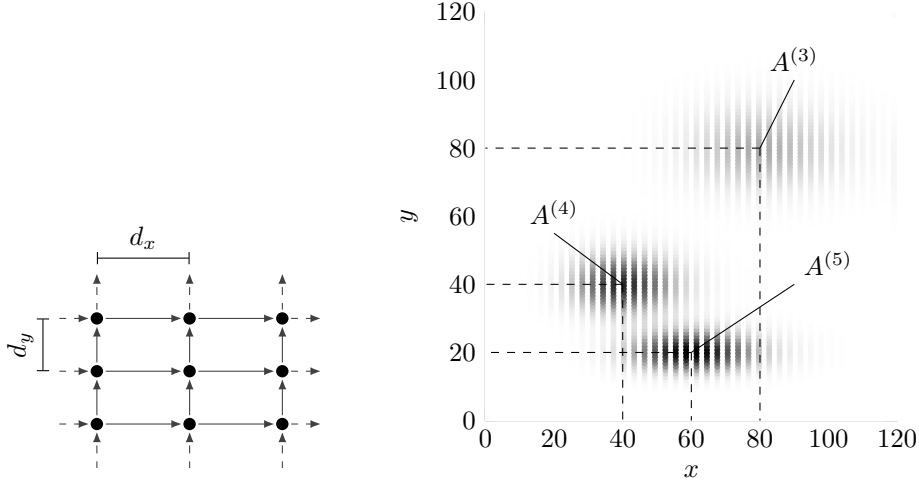

\begin{example}\label{example:motion_on_infinite_grids}
Let $G$ be an (infinite) oriented rectangular grid with nodes $d_x\mathbb{Z}\times d_y\mathbb{Z}$, where $d_x$ and $d_y$ are the width and height of each cell, respectively, as shown in Figure~\ref{fig:infinite-grid}. We aim to study the advection dynamics of a unit mass  initially concentrated at a node $v$. Let $A$ be an advection operator that satisfies the axioms of Locality (Axiom~\ref{axiom:locality}) and Advection~I (Axiom~\ref{axiom:advection-I}). These axioms ensure that the support of $A_G\1_v$ is $N^+[v]=\{v,v_x,v_y\}$, where $v_x$ and $v_y$ denote the two nodes to the right of $v$ and above $v$, respectively. Due to the symmetry of $G$, we find that the quantities $a_{vv}$, $a_{vv_x}$, and $a_{vv_y}$ are independent of~$v$, allowing us to define
\[
\alpha_0=a_{vv}, \quad \alpha_x=-a_{vv_x}, \quad \alpha_y=-a_{vv_y}.
\]

Let $u$ be a node with a horizontal distance of $n_xd_x$ from $v$ and a vertical distance of $n_yd_y$. The number of walks from $v$ to $u$ of length $k$ (i.e.,\, with $k-1$ edges) is
\[
[A_G^k\1_v](u) =
\alpha_0^{k-n_x-n_y}\alpha_x^{n_x}\alpha_y^{n_y}\binom{k}{n_x+n_y}\dfrac{(n_x+n_y)!}{n_x!n_y!}
\]
and therefore
\begin{equation}
\label{e:grid_motion_general_case}
f_t(u)=\bigl[e^{-tA_G}\1_v\bigr](u)=\sum_{k=0}^\infty \dfrac{(-t)^k}{k!} [A_G^k\1_v](u)=\dfrac{(t\alpha_x)^{n_x}(t\alpha_y)^{n_y}}{n_x!n_y!}e^{-\alpha_0t}.
\end{equation}

If $A$ also satisfies the axioms of Mass Transfer (Axioms~\ref{axiom:mass-transfer-I} and~\ref{axiom:mass-transfer-II}), then $\alpha_0=\alpha_x+\alpha_y$, so~\eqref{e:grid_motion_general_case} can be written in terms of Poisson distributions of the variables $n_x$ and $n_y$:
\[
f_t(u)=\dfrac{(t\alpha_x)^{n_x}e^{-t\alpha_x}}{n_x!}\cdot\dfrac{(t\alpha_y)^{n_y}e^{-t\alpha_y}}{n_y!}.
\]

The expected values of $n_x$ and $n_y$ are $t\alpha_x$ and $t\alpha_y$, respectively.
Assuming that the grid is embedded in $\mathbb{R}^2$, with origin in $v$, this corresponds to the point $(t\,d_x\alpha_x,\, t\,d_y\alpha_y)$. If $A$ also satisfies Advection II (Axiom~\ref{axiom:advection-II}) then we have $\alpha_xd_x+\alpha_yd_y=1$, so that the expected value of the mass distribution is $(ts,t(1-s))$ for some $0\le s\le1$. In this way, the mass has travelled, on average, a total distance of $t$.

Finally, if $A$ satisfies the axiom of Splitting  (Axiom~\ref{axiom:splitting}), then $\alpha_xd_x=\alpha_yd_y$, which implies $s=1/2$, and the expected value of the mass distribution is at $(t/2, t/2)$, consistent with constant-speed motion along the edges' directions.

Figure~\ref{f:grid_simulation} summarizes the behaviors described above, where the operators $A^{(3)}$, $A^{(4)}$, and $A^{(5)}$ are simulated numerically on a grid with $d_x = 3d_y$. As observed, only the solution $f_t$ corresponding to $A^{(4)}$ moves at the correct speed and direction.
\end{example}

\subsection{Advection in a half-line graph}
\label{ss:advection2}

While in Example~\ref{example:motion_on_infinite_grids} we studied advection in a 2-dimensional setting, here we focus on the $1$-dimensional case. Naturally, the advection process in a simple half-line graph could be computed analytically, yielding similar results to those in Example~\ref{example:motion_on_infinite_grids}. However, in this case, we modify the topology of the line graph by adding shortcut edges to evaluate how different operators handle the altered structure.

\begin{figure}
\centering
\begin{subfigure}{\textwidth}
\centering
\begin{tikzpicture}[xscale=1, yscale=.6,baseline, anchor=base]
\SetEdgeStyle[LineWidth=.4pt]
\Vertex[x=1,y=0, label=0, position=south, fontsize=\footnotesize, distance=-2]{B}
\Vertex[x=2,y=0, label=1, position=south, fontsize=\footnotesize, distance=-2]{C}
\Vertex[x=3,y=0, label=2, position=south, fontsize=\footnotesize, distance=-2]{D}
\Vertex[x=4,y=0, label=3, position=south, fontsize=\footnotesize, distance=-2]{E}
\Vertex[x=5,y=0, label=4, position=south, fontsize=\footnotesize, distance=-2]{F}
\Vertex[x=6,y=0, label=5, position=south, fontsize=\footnotesize, distance=-2]{G}
\Vertex[x=7,y=0, label=6, position=south, fontsize=\footnotesize, distance=-2]{H}
\Vertex[x=8,y=0, label=7, position=south, fontsize=\footnotesize, distance=-2]{I}
\Vertex[x=9,y=0, color=white]{L}
\Edge[Direct](B)(C)
\Edge[Direct](C)(D)
\Edge[Direct](D)(E)
\Edge[Direct](E)(F)
\Edge[Direct](F)(G)
\Edge[Direct](G)(H)
\Edge[Direct](H)(I)
\Edge[Direct, style=dashed](I)(L)
\Edge[Direct, bend=55](B)(D)
\Edge[Direct, bend=55](C)(E)
\Edge[Direct, bend=55](D)(F)
\Edge[Direct, bend=55](E)(G)
\Edge[Direct, bend=55](F)(H)
\Edge[Direct, bend=55](G)(I)
\Edge[Direct, bend=55, style=dashed](H)(L)
\end{tikzpicture}
\caption{The half-line graph $G$ with the addition of shortcut edges. Every edge from a node $u$ to $u+1$ has length $1$, while every edge from $u$ to $u+2$ has length $2$.}
\label{f:half_line_graph}
\end{subfigure}
\begin{subfigure}{\textwidth}
\centering
\begin{tikzpicture}
\begin{axis}[
    xlabel={$u$ (node)},
    ylabel={$f_{150}(u)$},
    ylabel style = {rotate = -90},
    width=12cm,
    height=5.5cm,
    ymin=0,
    ymax=0.045,
    xmin=0,
    xmax=400,
    xticklabel style={
        font=\small,
    },
    yticklabel style={
        font=\small,
        /pgf/number format/fixed,
    },
    legend pos = outer north east,
]
\addplot [very thick] table {figures/advection_A2.txt}
coordinate[pos=0.7,pin={[pin edge={thick}, black]40:\small$A^{(2)}$ (scaled)}, fill=black] ();
\addplot [very thick] table {figures/Advection_A3.txt}
coordinate[pos=0.7,pin={[pin edge={thick}, black]160:\small$A^{(3)}$}] ();
\addplot [very thick] table {figures/Advection_A4.txt}
coordinate[pos=0.375,pin={[pin edge={thick}, black]60:\small$A^{(4)}$}, fill=black] ();
\addplot [very thick] table {figures/Advection_A5.txt}
coordinate[pos=0.375,pin={[pin edge={thick}, black]20:\small$A^{(5)}$}, fill=black] ();
\end{axis}
\end{tikzpicture}
\caption{Plot of $f_t=e^{-A_Gt}f_0$ according to different operators, with $t=150$ and initial unitary mass concentrated at node $0$. In the simulation, the graph has been truncated at node $400$. The plot relative to $A^{(2)}$ has been scaled down to $4\%$.}
\label{f:half_line_results}
\end{subfigure}
\caption{Advection motion on an infinite half-line graph with shortcut edges.}
\label{f:half_line}
\end{figure}
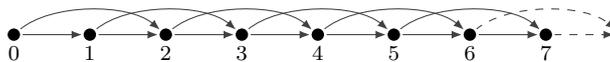
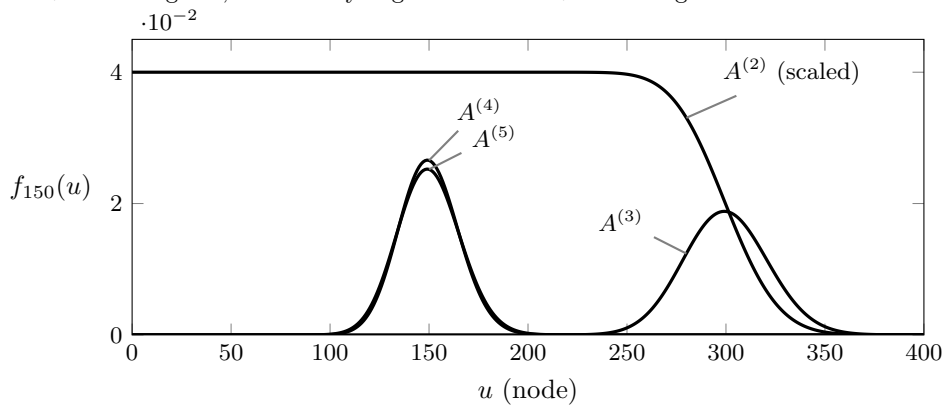

\begin{example}
We consider the infinite half-line graph $G$ of Figure~\ref{f:half_line_graph}, with set of nodes $\mathbb{N}$: every node $u$ is connected to $u+1$ with an edge of length $1$ and to $u+2$ with an edge of length $2$. Note that $G$ belongs to the class of graphs described in Remark~\ref{r:extension_coherent_distances}, as we can set $\phi(u)=u$. Figure~\ref{f:half_line_results} shows the results of numerical simulation of the advection equation at $t=150$, with initial mass on node~$0$, according to the operators defined in Example~\ref{example:four-advections}.

As expected, all the operators preserve the positivity of the solution. However, the simulation with $A^{(2)}$ fails entirely, as part of the mass becomes trapped in the visited nodes. Both $A^{(4)}$ and $A^{(5)}$ satisfy the axiom of Advection II (Axiom~\ref{axiom:advection-II}), ensuring that the mass moves at the correct speed, with an average distance of exactly 150 from node~$0$ (see also Remark~\ref{r:extension_coherent_distances}).
\end{example}

\subsection{Mass splitting on a tree}
\label{ss:splitting}

In this section, we consider another advection example with focus on the axiom of Splitting. A small tree is considered for convenience, though larger (or even infinite) trees would produce similar results.

\begin{figure}
\begin{subfigure}{\textwidth}
\centering
\begin{tikzpicture}
\begin{scope}[xscale=1.5, yscale=1.5]
\tikzgraphsettings
\node (0) [label={[label distance=1]90:$v$}] at (0,0) {};
\node (1) [label={[label distance=0]90:$z_1$}] at (0:1) {};
\node (11) [label={[label distance=0]90:$z_2$}] at (0:1+0.5) {};
\node (111) [label={[label distance=0]80:$z_3$}] at (0:1+0.5+.25) {};
\node (2) [label={[label distance=0]80:$w_1$}] at (150:0.5) {};
\node (22) [label={[label distance=0]80:$w_2$}] at (150:0.5+0.5) {};
\node (222) [label={[label distance=0]100:$w_3$}] at (150:0.5+0.5+.25) {};
\node (3) [label={[label distance=0]-80:$u_1$}] at (210:0.33) {};
\node (33) [label={[label distance=0]-80:$u_2$}] at (210:0.33+0.5) {};
\node (333) [label={[label distance=0]-100:$u_3$}] at (210:0.33+0.5+.25) {};
\draw (0) [->] to (1);
\draw (0) [->] to (2);
\draw (0) [->] to (3);
\draw (1) [->] to (11);
\draw (11) [->] to (111);
\draw (2) [->] to (22);
\draw (22) [->] to (222);
\draw (3) [->] to (33);
\draw (33) [->] to (333);
\end{scope}
\end{tikzpicture}
\caption{A tree $T\in\mathcal{G}$ such that $d_{vz_1}=2\,d_{vw_1}=3\,d_{vu_1}$.}
\label{f:tree}
\end{subfigure}

\begin{subfigure}{\textwidth}
\centering
\parbox[t]{6.5cm}{\null\par
\raggedleft\begin{tikzpicture}
\begin{axis}[
    width=5cm,
    height=3.8cm,
    title={\small$A^{(2)}$},
    legend style={
        at={(-.85,0.246)},
        anchor= west,
        row sep=0ex},
    xticklabel style={font=\small}, yticklabel style={font=\small},
]
\addplot [ultra thick, dotted, color=mplblue] table {figures/tree_splitting_2_branch_1.txt};
\addplot [ultra thick, dashed, color=mplorange] table {figures/tree_splitting_2_branch_2.txt};
\addplot [ultra thick, color=mplgreen] table {figures/tree_splitting_2_branch_3.txt};
\addlegendentry{\raisebox{10pt}{\footnotesize$\displaystyle\sum_{i=1}^3f_t(z_i)$}}
\addlegendentry{\raisebox{10pt}{\footnotesize$\displaystyle\sum_{i=1}^3f_t(w_i)$}}
\addlegendentry{\raisebox{10pt}{\footnotesize$\displaystyle\sum_{i=1}^3f_t(u_i)$}}
\end{axis}
\end{tikzpicture}}
\parbox[t]{4.2cm}{\null\par\begin{tikzpicture}
\begin{axis}[
    width=5cm,
    height=3.8cm,
    title={\small$A^{(3)}$},
    xticklabel style={font=\small}, yticklabel style={font=\small},
]
\addplot [ultra thick, dotted, color=mplblue] table {figures/tree_splitting_3_branch_1.txt};
\addplot [ultra thick, dashed, color=mplorange] table {figures/tree_splitting_3_branch_2.txt};
\addplot [ultra thick, color=mplgreen] table {figures/tree_splitting_3_branch_3.txt};
\end{axis}
\end{tikzpicture}}

\vspace{-.5cm}
\parbox[t]{6.5cm}{\raggedleft\begin{tikzpicture}
\begin{axis}[
    width=5cm,
    height=3.8cm,
    title={\small$A^{(4)}$},
    xticklabel style={font=\small}, yticklabel style={font=\small},
]
\addplot [ultra thick, dotted, color=mplblue] table {figures/tree_splitting_4_branch_1.txt};
\addplot [ultra thick, dashed, color=mplorange] table {figures/tree_splitting_4_branch_2.txt};
\addplot [ultra thick, color=mplgreen] table {figures/tree_splitting_4_branch_3.txt};
\end{axis}
\end{tikzpicture}}
\parbox[t]{4.2cm}{\begin{tikzpicture}
\begin{axis}[
    width=5cm,
    height=3.8cm,
    title={\small$A^{(5)}$},
    xticklabel style={font=\small}, yticklabel style={font=\small},
]
\addplot [ultra thick, dotted, color=mplblue] table {figures/tree_splitting_5_branch_1.txt};
\addplot [ultra thick, dashed, color=mplorange] table {figures/tree_splitting_5_branch_2.txt};
\addplot [ultra thick, color=mplgreen] table {figures/tree_splitting_5_branch_3.txt};
\end{axis}
\end{tikzpicture}}
\caption{Plot of the mass in each branch of $T$ for $t\in[0,1]$, with initial mass at $v$.}
\label{f:tree_simul}
\end{subfigure}
\caption{Advection process on a small tree.}
\end{figure}
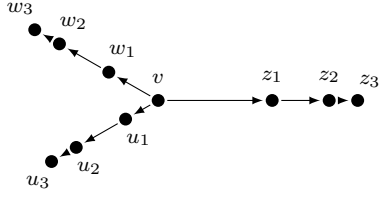
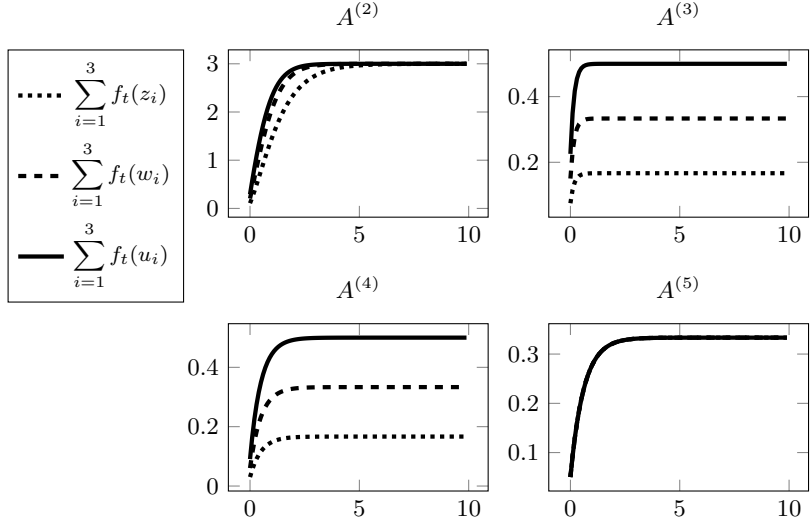

\begin{example}
Let $T$ be the tree in Figure~\ref{f:tree}, such that the distances from the root $v$ to its children $u_1$, $w_1$ and $z_1$ are $1/3$, $1/2$ and $1$, respectively. Let us consider an advection process $f_t$ with initial unit mass concentrated at node $0$, and define $s_u(t) = f_t(u_1) + f_t(u_2) + f_t(u_3)$, and similarly for $s_w(t)$ and $s_z(t)$.

The axiom of Splitting (Axiom~\ref{axiom:splitting}) states that, in the limit, $s_u(t)$, $s_w(t)$, and $s_z(t)$ are in the ratio $3:2:1$. Figure~\ref{f:tree_simul} shows the results of numerical simulations of the advection process using different operators. According to Corollary~\ref{c:splitting}, both $A^{(3)}$ and $A^{(4)}$ satisfy the axiom of Splitting, and indeed the mass divides among the three branches in the correct proportions. The other two operators exhibit a clearly different behavior: $A^{(5)}$ was designed so that $a_{uv}$ is constant for $u\in N^+(v)$, leading to an equal division of mass among $v$'s children. On the other hand, $A^{(2)}$  satisfies the conditions of Corollary~\ref{c:splitting}, but the result is not applicable as its hypotheses are not met, and the Splitting Axiom is not satisfied.
\end{example}

\begin{remark}
Note that in the case of $A^{(3)}$, $A^{(4)}$ and $A^{(5)}$, it also seems that the ratio of mass in each branch is constant over time. This is a consequence of~\eqref{e:exact_formula}, in conjunction with the fact that  $a_{u_iu_i}=a_{w_iw_i}=a_{z_iz_i}$ for $i=1,2,3$, as $N^+[u_i]$, $N^+[w_i]$, and $N^+[z_i]$ are isomorphic.
\end{remark}

\subsection{Advection on road networks}

We conclude this series of examples with a simulation of the advection process on a real-world graph.

\begin{figure}[p]
\centering
\begin{subfigure}{.5\textwidth}
\includegraphics[width=\textwidth]{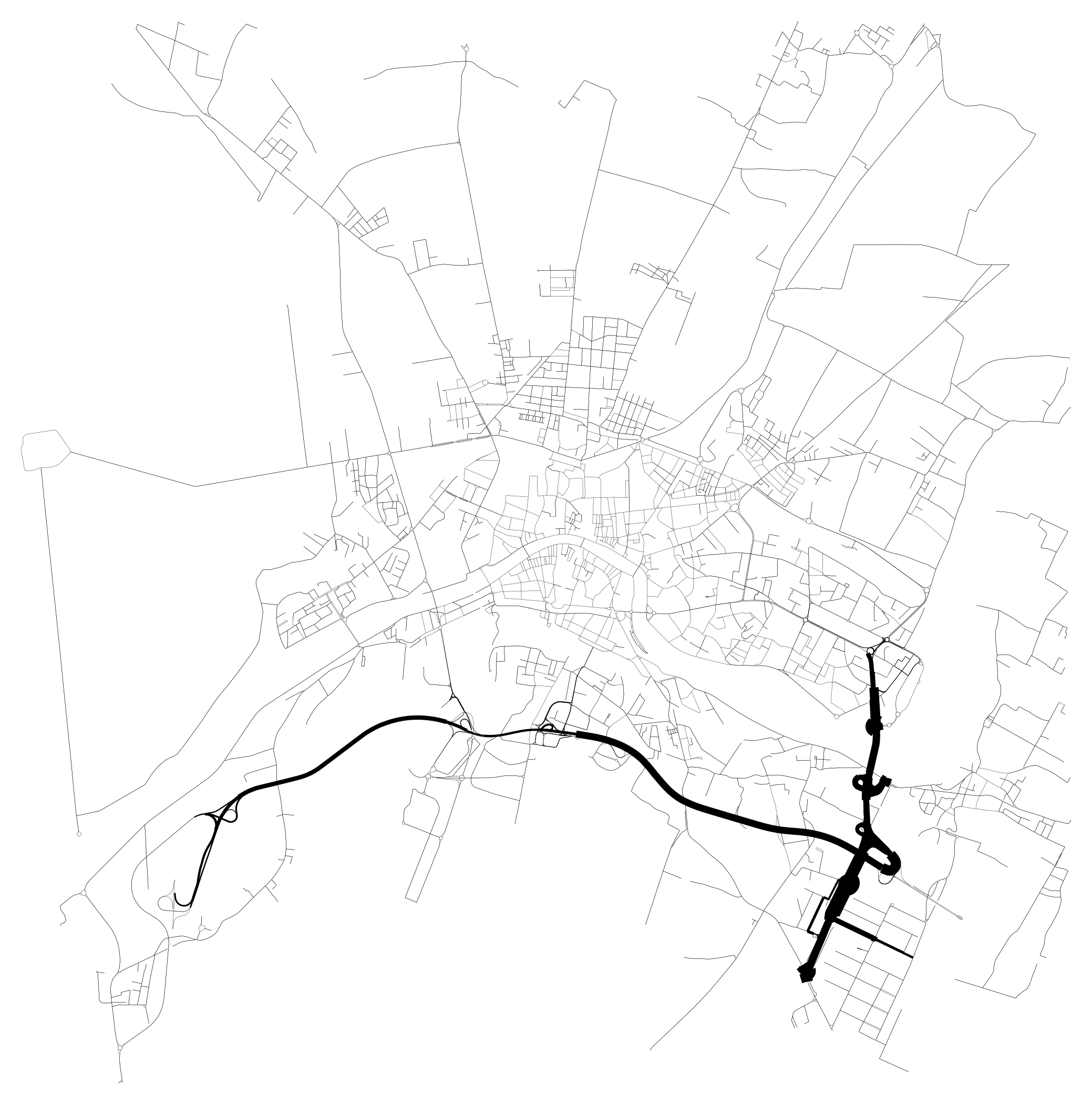}%
\llap{\includegraphics[width=\textwidth]{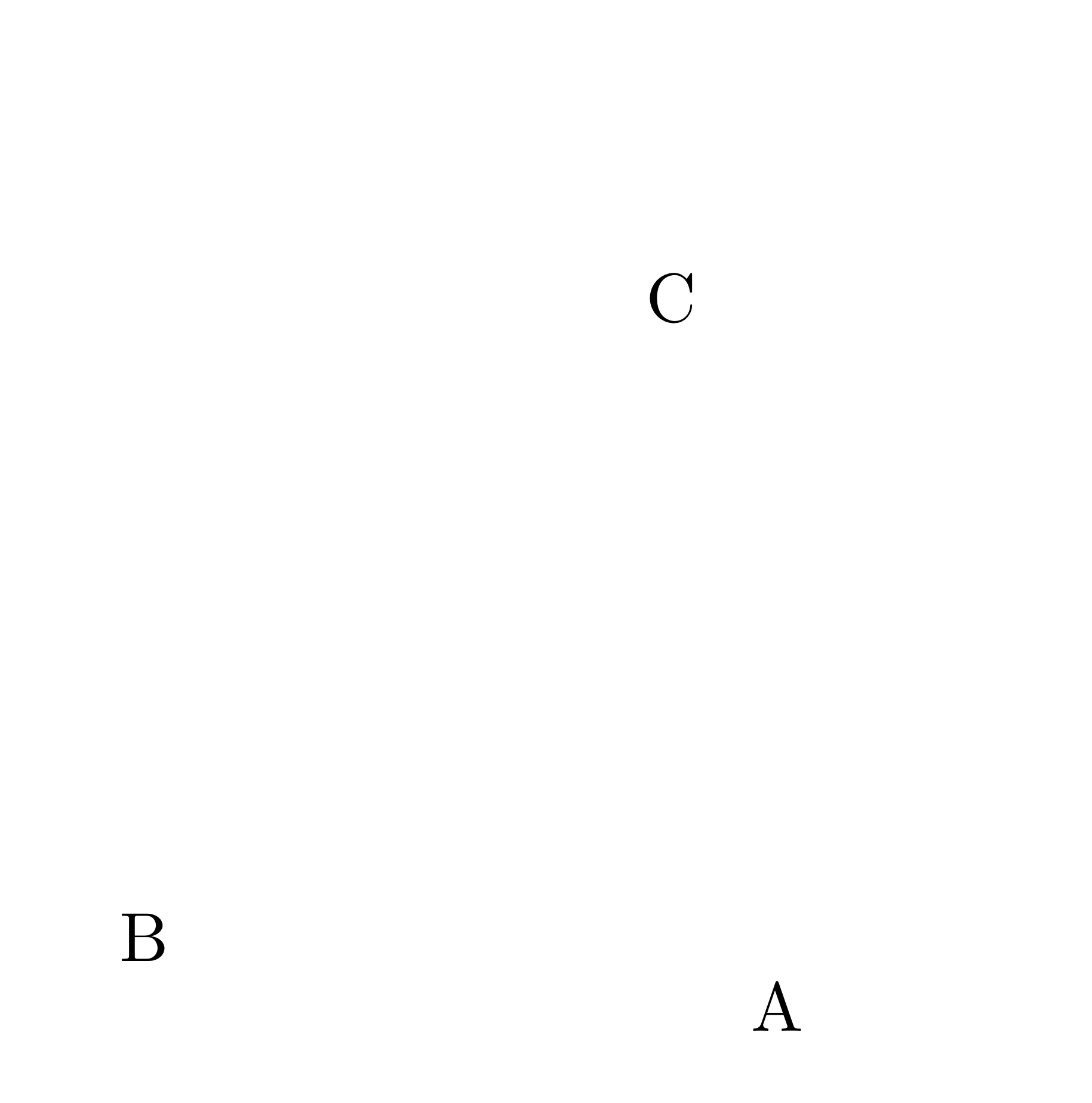}}
\caption{$t=500$}
\label{500}
\end{subfigure}%
\begin{subfigure}{.5\textwidth}
\includegraphics[width=\textwidth]{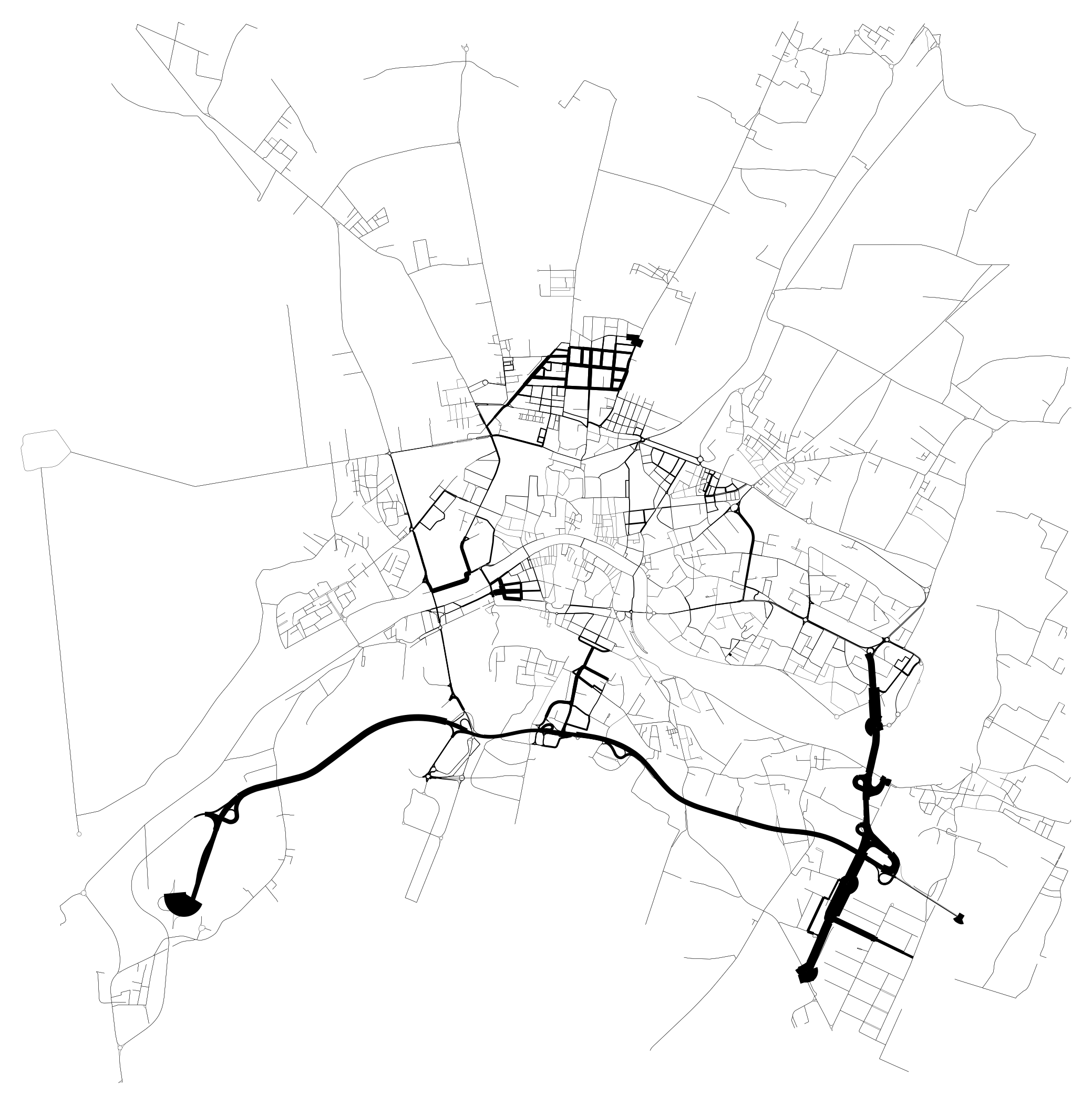}%
\llap{\includegraphics[width=\textwidth]{abc.png}}	
\caption{$t=5000$}
\label{5000}
\end{subfigure}

\begin{subfigure}{.5\textwidth}
\includegraphics[width=\textwidth]{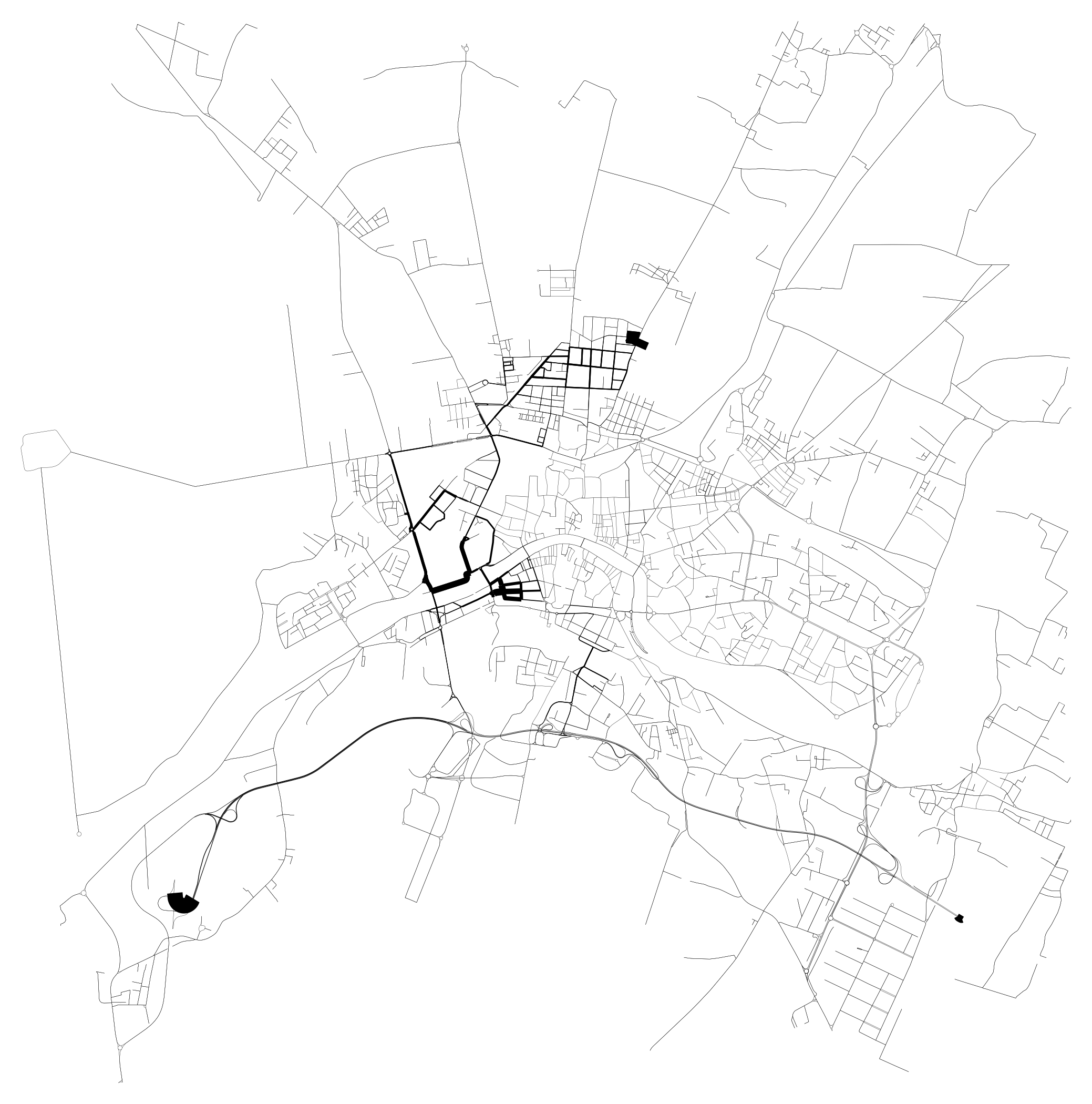}%
\llap{\includegraphics[width=\textwidth]{abc.png}}		
\caption{$t=25000$}
\label{25000}
\end{subfigure}%
\begin{subfigure}{.5\textwidth}
\includegraphics[width=\textwidth]{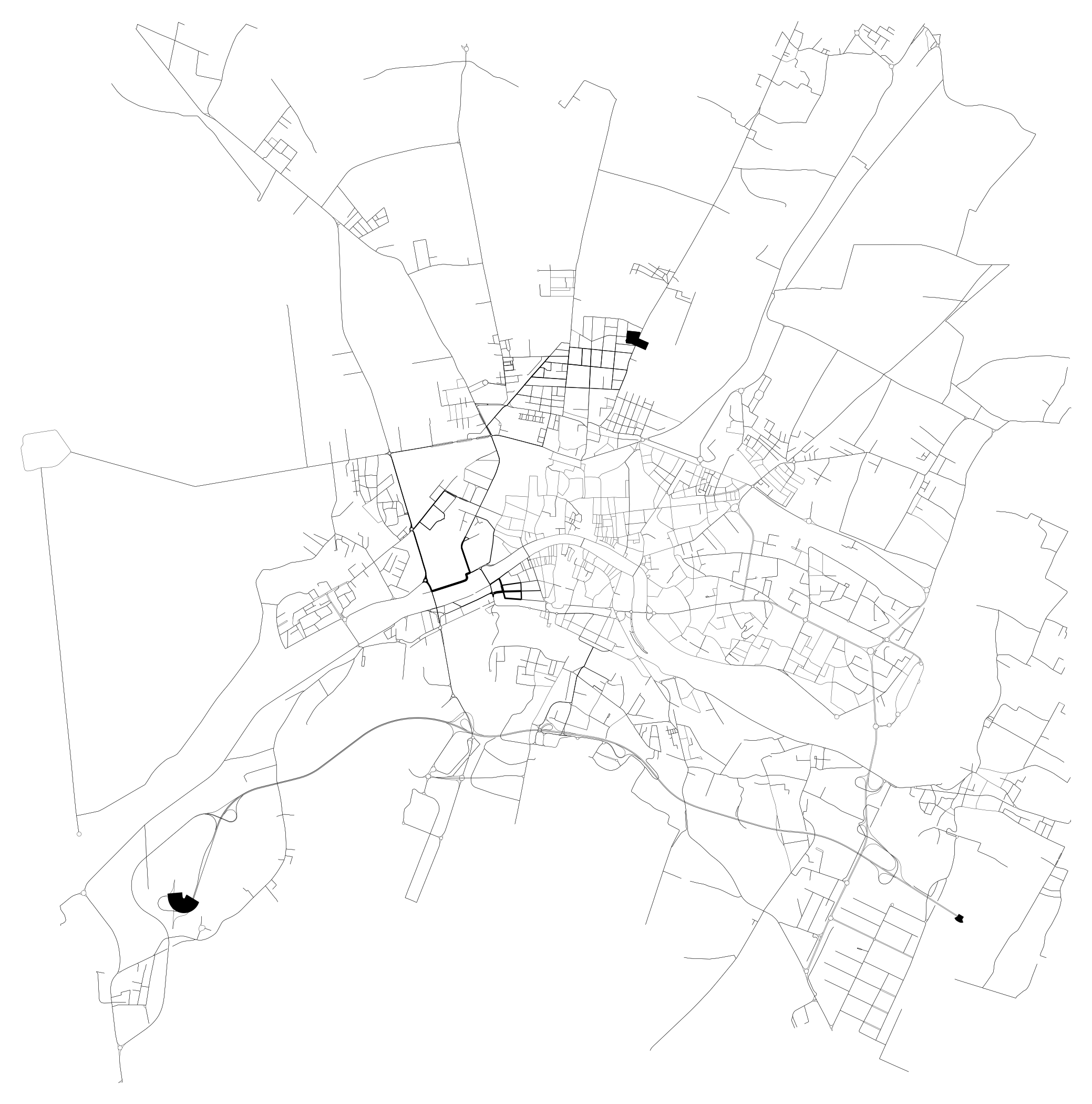}%
\llap{\includegraphics[width=\textwidth]{abc.png}}
\caption{$t=50000$}
\label{50000}
\end{subfigure}
\caption{Advection process on the street network of Pisa. The initial mass is concentrated at node $A$ and directed to nodes $B$ and~$C$. The traffic concentration is highlighted by the means of edge widths. The width of each edge $(u,v)$ is $\alpha+\beta \sqrt{f_t(u)f_t(v)}$, where $\alpha$ and $\beta$ were chosen for best readability.}
\label{f:pisa}
\end{figure}

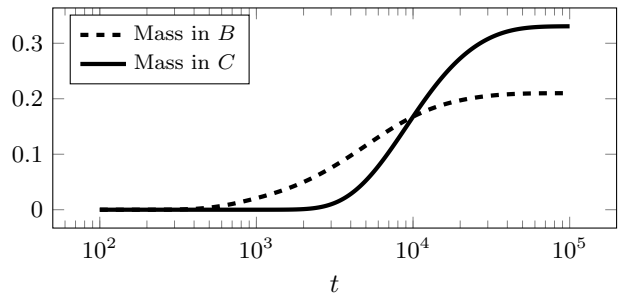
\begin{figure}[p]
\null\bigskip
\begin{tikzpicture}
\begin{axis}[
    xmode=log,
    width=9cm,
    height=4.5cm,
    legend pos = north west,
    xlabel=$t$,
    xticklabel style={font=\small},
    yticklabel style={font=\small},
]
\addplot [ultra thick, dashed, color=mplblue] table {figures/mass_in_B.txt};
\addplot [ultra thick, color=mplorange] table {figures/mass_in_C.txt};
\addlegendentry{\footnotesize Mass in $B$}
\addlegendentry{\footnotesize Mass in $C$}
\end{axis}
\end{tikzpicture}
\par\vspace{-.5\baselineskip}
\caption{Time evolution of the mass in nodes $B$ and $C$, i.e., $f_t(B)$ and $f_t(C)$, according to the dynamic described in Figure~\ref{f:pisa}.}
\label{fig:pisa_mass_bc}
\end{figure}

\begin{example}
Let $G=(V,E, \omega)$ be the directed graph representing the roads in a $10\,$km box around the city of Pisa obtained through the \textsc{OSMNx} package \cite{OSMnx}, where the weights $\omega$ represent the \emph{travel time} of the roads. Let $A$ be a node of the graph (Figure~\ref{f:pisa}). As most roads are bidirectional, $G$ is almost an undirected graph, and the advection process with initial mass in $A$ would result in a diffusion-like process (see Remark~\ref{r:laplacian}). Instead, we aim to model traffic motion from $A$ towards two target nodes~$B$ and~$C$. To enforce this behavior, we construct two sets of directed edges $E_B$ and $E_C$ as follows. We start by collecting all the \emph{unidirectional} edges of $G$ into~$E_B$, while for any \emph{bidirectional} edge $(u,v)$ of $G$, we add $(u,v)$ to $E_B$ if $v$ is closer to~$B$ than~$u$; otherwise, we add $(v,u)$. The same construction applies for $E_C$, so that $E_B$ and $E_C$ encode the preferred orientations of $G$'s edges in order to reach $B$ and~$C$, respectively. We combine the two orientations in a single directed graph, setting
\[
G_{BC} = \left(V,\, E_B\cup E_C,\, \restr\omega{E_B\cup E_C}\right)\!.
\]
Note that each edge of $G$ can be traversed in $G_{BC}$ in at least one direction. Also, we turn $B$ and $C$ into \emph{sink nodes} by removing all outgoing edges, ensuring that mass flow cannot escape from them (Axiom~\ref{axiom:advection-I}). Then we choose the operator characterized in Theorem~\ref{t:characterization} (i.e., $A^{(4)}$) to build the advection operator~$A_{G_{BC}}$. Figure~\ref{f:pisa} shows the solution of the differential equation
\[
\dfrac{\de }{\de t} f_t = -A_{G_{BC}} f_t
\]
for different time frames, while Figure~\ref{fig:pisa_mass_bc} displays the mass accumulated at nodes $B$ and $C$ over time.

We note that the mass follows the two-target orientation, splitting into two branches directed towards $B$ and $C$ (Figure~\ref{500}). In this initial phase, although $B$ and $C$ are at comparable distances from $A$, the mass reaches $B$ faster than $C$, as the path to $B$ is more direct (Figure~\ref{fig:pisa_mass_bc}). The reticular structure between $A$ and $C$ causes the flow towards $C$ to disperse more. As the dynamic evolves, part of the mass accumulates in the central zone of Pisa near $C$ (Figure~\ref{5000}). In this transient phase, some of the mass initially directed toward $B$ diverts toward $C$ through the city center (and vice versa). At time $t = 25000$, the flow has almost completely left~$A$, concentrating around $B$, $C$, and along the streets connecting them, where some roads are more congested than others (Figure~\ref{25000}). Eventually, after a sufficiently long time, almost all the mass concentrates at $B$ and $C$ (Figure~\ref{50000}). Interestingly, at this point, there is more mass at $C$ than at $B$ (Figure~\ref{fig:pisa_mass_bc}). This inversion may be explained as follows: as noted earlier, the reticular structure around $C$ initially slows the advection towards $C$. However, this intricate network of streets acts as a reserve, incidentally capturing mass that diverts from both $C$ and $B$. Most of the mass in this reserve will eventually, albeit slowly, flow to the closer target $C$.
\end{example}

\section{Conclusions and future perspectives}
\label{sec:conclusion}

In this paper we have developed a comprehensive axiomatic framework for characterizing advection operators on oriented, distance-weighted graphs. By establishing key properties such as locality, mass transfer, and proper advection behavior, we have given a rigorous definition of what an operator that model transport dynamics on discrete networks should look like.

The axioms and their corresponding characterizations offer valuable insights, linking the choice of parameters in the definition of the operators to specific properties of the induced dynamics. Among the discussed operators, only the one satisfying all the axioms provides an accurate representation of mass movement, including conservation and directional flow. Through examples on both finite and infinite dimensional graphs and numerical simulation we have illustrated the applicability and potential use in real-world scenarios like traffic flow modeling.

The framework opens avenues for further exploration, including extensions to more general graphs and variations of the advection operator that may account for variable velocity or additional dynamics, e.g., considering the interplay with local and nonlocal diffusion, other phenomena with memory or the relaxation of the locality assumptions to encompass larger neighborhoods. We also plan to investigate the link between continuous dynamics and the scaling limit of the model proposed here. Applications of interest for these extensions may include the study of biological phenomena, particularly epidemiological processes within networks \cite{BellomoNetworkCovid,bertaglia2024pandemics}.

\section*{\refname}
\printbibliography[heading=none]

\end{document}